%% file: DM834.tex
\documentclass[conference]{IEEEtran}
\IEEEoverridecommandlockouts
\usepackage{cite}
\usepackage{tikz}
\usepackage{booktabs}
\usepackage{amsmath,amssymb,amsfonts}
\usepackage{bbm} 
\usepackage{algorithmic}
\usepackage{graphicx}
\usepackage{textcomp}
\usepackage{array}
\usepackage{xcolor}
\def\BibTeX{{\rm B\kern-.05em{\sc i\kern-.025em b}\kern-.08em
    T\kern-.1667em\lower.7ex\hbox{E}\kern-.125emX}}
\usepackage{xspace}
\usepackage{amsfonts}
\usepackage{mathtools}
\usepackage{amsfonts}
\usepackage{balance}
\usepackage[ruled, vlined, nofillcomment, linesnumbered]{algorithm2e}
\usepackage{amsthm}
\usepackage[square,sort,comma,numbers]{natbib}
\usepackage{bbm} 
\usepackage[linesnumbered]{algorithm2e}

\input{defines}

\usepackage[bookmarks=false]{hyperref}

\begin{document}

\title{Robust Cascade Reconstruction by\\ Steiner Tree Sampling}


\author{
\IEEEauthorblockN{Han Xiao}
\IEEEauthorblockA{\textit{Aalto University} \\
Espoo, Finland \\
han.xiao@aalto.fi}
\and
\IEEEauthorblockN{Cigdem Aslay}
\IEEEauthorblockA{\textit{Aalto University} \\
Espoo, Finland \\
cigdem.aslay@aalto.fi}
\and
\IEEEauthorblockN{Aristides Gionis}
\IEEEauthorblockA{\textit{Aalto University} \\
Espoo, Finland \\
aristides.gionis@aalto.fi}
}

\maketitle

\begin{abstract}
\input{abstract}
\end{abstract}


\input{intro}
\label{sec:intro}

\input{related}

\label{sec:related}

\input{problem}

\label{sec:problem}

\input{sampling}

\label{sec:sampling}

\input{experiment}
\label{sec:experiment}

\input{conclusion}

\label{sec:conclusion}


\section*{Acknowledgments}
This work has been supported by three Academy of Finland projects  (286211, 313927, and 317085), 
and the EC H2020 RIA project ``SoBigData'' (654024)". Part of the work was done while the second author was at ISI Foundation.

{
\balance
\bibliographystyle{IEEEtran}
\bibliography{references}
}

\end{document}

%% file: defines.tex
 \newtheorem{problem}{Problem}
 \newtheorem{theorem}{Theorem}

 \newtheorem{lemma}{Lemma}

 \newcommand{\path}{\ensuremath{\mathit{path}}\xspace}

\newcommand{\spara}[1]{\smallskip\noindent{\bf #1}}

\newcommand{\graph}{\ensuremath{G}\xspace}
\newcommand{\nodes}{\ensuremath{V}\xspace}

\newcommand{\tree}{\ensuremath{T}\xspace}

\newcommand{\bigO}{\ensuremath{\mathcal{O}}\xspace}

\newcommand{\sharpP}{\ensuremath{\#\mathbf{P}}\xspace}
\newcommand{\set}[1]{\left\{#1\right\}}

\newcommand{\abs}[1]{{\left|#1\right|}}

\newcommand{\proba}[1]{\text{Pr}\left({#1}\right)\xspace}

\newcommand{\algcp}{\ensuremath{\textsc{\large cycle\-pop\-ping}}\xspace}
\newcommand{\algcut}{\ensuremath{\textsc{\large trim}}\xspace}
\newcommand{\alglerw}{\ensuremath{\textsc{\large lerw}}\xspace}

\newcommand{\terminals}{\ensuremath{X}\xspace}

\newcommand{\obs}{\ensuremath{O}\xspace}
\newcommand{\obsuninfected}{\ensuremath{U}\xspace}

\newcommand{\treeroot}{\ensuremath{r}}

\newcommand{\laplacian}{\ensuremath{L}\xspace}

\newcommand{\rsptree}{\ensuremath{T}\xspace} 
\newcommand{\rsttree}{\ensuremath{\tree_X}} 
\newcommand{\rsptrees}{\ensuremath{\mathcal{S}}\xspace}
\newcommand{\rsptreest}{\ensuremath{\rsptrees_{\rsttree}}\xspace}
\newcommand{\parExp}[2]{\underset{#1}{\mathbb{E}}\left[#2\right]\xspace} 
\newcommand{\cg}{\ensuremath{\chainG_c}\xspace}
\newcommand{\cgraph}{\ensuremath{\cg[\rsttree]}\xspace} 
\newcommand{\cnodes}{\ensuremath{V_c}\xspace} 
\newcommand{\cedges}{\ensuremath{\tilde{E}_c}\xspace} 
\newcommand{\outNeighbor}[1]{\tilde{N}_{\mathit{out}}(#1)\xspace} 
\newcommand{\inNeighbor}[1]{\tilde{N}_{\mathit{in}}(#1)\xspace} 
\newcommand{\supernode}{\ensuremath{v_{c}}\xspace}
\newcommand{\chainG}{\ensuremath{\tilde{G}}\xspace} 
\newcommand{\chainE}{\ensuremath{\tilde{E}}\xspace} 
\newcommand{\Nin}{\ensuremath{N_{\mathit{in}}}\xspace} 
\newcommand{\Nout}{\ensuremath{N_{\mathit{out}}}\xspace} 

\SetKwComment{Comment}{$\triangleright$\ }{}
\newcommand{\nextnode}{\ensuremath{\text{Next}}}
\newcommand{\intree}{\ensuremath{\text{InTree}}}

\newcommand{\ourmethod}{\textit{Tree-sampling}\xspace}
\newcommand{\pagerank}{\textit{PageRank}\xspace}
\newcommand{\mintree}{\textit{Min-Steiner-tree}\xspace}
\newcommand{\netfill}{\textit{NetFill}\xspace}

\definecolor{darkred}{rgb}{0.5793618007033479,0.042445213537590176,0.07361784131795752}
\definecolor{myblue}{rgb}{0.12156862745098039, 0.47058823529411764, 0.7058823529411765}
\definecolor{myyellow}{rgb}{1.0, 0.8509803921568627, 0.1843137254901961}

\newcommand{\truepos}{\tikz\draw[darkred,fill=darkred] (0,0) circle (0.76ex);\xspace}

\newcommand{\falsepos}{\tikz\draw[darkred,fill=darkred] (0,0) circle (0.5ex);\xspace}
\newcommand{\falseneg}{\tikz\draw[gray,fill=white] (0,0) circle (0.76ex);\xspace}

\newcommand{\obsnode}{\tikz\draw[gray,fill=myblue] (0,0) rectangle (1.5ex, 1.5ex);\xspace}
\newcommand{\hiddeninfnode}{\tikz\draw[gray,fill=myyellow] (0,0) circle (0.76ex);\xspace}
\newcommand{\uninfnode}{\tikz\draw[gray,fill=white] (0,0) circle (0.5ex);\xspace}

\newcommand{\ap}{\ensuremath{\mathit{AP}}\xspace}

\newcommand{\grqc}{{\texttt{grqc}}\xspace}
\newcommand{\infectious}{{\texttt{infectious}}\xspace}
\newcommand{\email}{{\texttt{email-univ}}\xspace}
\newcommand{\student}{{\texttt{student}}\xspace}

\newcommand{\lattice}{{\texttt{lattice}}\xspace}
\newcommand{\digg}{{\texttt{Digg}}\xspace}

\newcommand{\ic}{\textit{IC}\xspace}
\newcommand{\si}{\textit{SI}\xspace}
\newcommand{\sir}{\textit{SIR}\xspace}

\newcommand{\mindistroot}{\textit{Min-dist}\xspace}
\newcommand{\pagerankroot}{\textit{Page\-Rank}\xspace}
\newcommand{\trueroot}{\textit{True-root}\xspace}

\makeatletter
\newcommand{\nosemic}{\renewcommand{\@endalgocfline}{\relax}}
\newcommand{\dosemic}{\renewcommand{\@endalgocfline}{\algocf@endline}}
\@ifpackageloaded{algorithm2e}{
	\SetKwComment{tcpas}{\{}{\}}
	\SetCommentSty{textnormal}
	\SetKw{Null}{null}
	\SetKw{AND}{and}
	\SetKw{OR}{or}
	\SetKw{EACH}{each}
	\SetKwInOut{Input}{input}
	\SetKwInOut{Output}{output}
	\SetKwInput{KwInput}{Input}
	\SetKwInput{KwOutput}{Output}
	\SetArgSty{textnormal}
	\SetKw{False}{false}
	\SetKw{True}{true}
	\SetKw{Break}{break}
	\SetKw{None}{none}
	\SetKw{Continue}{continue}
}{}
\makeatother

\newenvironment {squishlist}
{\begin{list}{$\bullet$}
		{ \setlength{\itemsep}{0pt}
			\setlength{\parsep}{3pt}
			\setlength{\topsep}{3pt}
			\setlength{\partopsep}{0pt}
			\setlength{\leftmargin}{1.5em}
			\setlength{\labelwidth}{1em}
			\setlength{\labelsep}{0.5em} } }
	{\end{list}}

%% file: abstract.tex
We consider a network where an infection has taken place
and a subset of infected nodes has been partially observed.
Our goal is to reconstruct the underlying cascade that is likely to have generated these observations. 
We reduce this cascade-reconstruction problem to computing the marginal probability 
that a node is infected given the partial observations, which is a \sharpP-hard problem. 
To circumvent this issue, we resort to estimating infection probabilities 
by generating a sample of probable cascades, which span the nodes that have already been observed to be infected, 
and avoid the nodes that have been observed to be uninfected.
The sampling problem corresponds to sampling directed Steiner trees with a given set of terminals, 
which is a problem of independent interest 
and has received limited attention in the literature.
For the latter problem we propose two novel algorithms with provable guarantees 
on the sampling distribution of the returned Steiner trees.

The resulting method improves over state-of-the-art approaches that often make explicit assumptions
about the infection-propagation model, or require additional parameters.
Our method provides a more robust approach to the cascade-reconstruction problem, 
which makes weaker assumptions about the infection model, 
requires fewer additional parameters, 
and can be used to estimate node infection probabilities. We experimentally validate the proposed reconstruction algorithm on real-world graphs with both synthetic and real cascades. We show that our method outperforms all other baseline strategies in most cases.

%% file: intro.tex
\section{Introduction}

Diffusion processes have been  used to model the spread of an item in a network, 
such as rumors, trends, or infections.
The study of diffusion processes has a wide range of applications, 
from designing strategies for viral marketing to preventing and controlling outbreaks.
In this paper we consider the problem of reconstructing a cascade, 
based on the partially observed outcome of a propagation process in a network. 

Reconstructing a cascade, which may have resulted from the propagation of an infection or information, has many important applications, and thus, 
several variants of the problem have been studied in the literature~\cite{rozenshtein2016reconstructing, braunstein2016network, sundareisan2015hidden, xiao2018reconstructing}. In this work, we consider the task of {\em probabilistic cascade reconstruction}. Given the partially observed infection state of the network, our aim is to reconstruct the infection cascade by inferring the hidden infection state of the nodes with no observation. We map our cascade reconstruction problem to the task of computing the marginal probability that a node is infected given partial observations. We show that computing these marginal probabilities is \sharpP-hard and resort to approximation by sampling cascades from a target distribution. The sampled cascades are required to span the nodes that have been observed to be infected, and avoid the nodes that have been observed to be uninfected. We map this sampling problem to the problem of sampling directed Steiner trees with a given set of terminals. 
The latter problem is of independent interest, and to our knowledge 
it has received limited attention in the literature.

Our main technical contribution is two novel algorithms with provable guarantees 
on the distribution of the sampled Steiner trees.
The first algorithm, \algcut, uses Wilson's algorithm~\cite{wilson1996generating} to sample spanning trees then ``trims'' the branches that are beyond the terminal nodes, with proper adjustment of the sampling probabilities so that the resulting Steiner tree is sampled from the desired target distribution. 
The second algorithm, loop-erased random walk (\alglerw), 
is an adaptation of Wilson's algorithm~\cite{wilson1996generating} adjusted to take into account the input set of terminal nodes and directly return a Steiner tree from the desired target distribution. 

\newlength{\figwidth}
\begin{figure}[t!]
  \centering
  \newcommand{\includefig}[1]{\includegraphics[width=0.31\linewidth]{figs/overview/#1}}
  \setlength{\figwidth}{4.0cm}
  \setlength{\tabcolsep}{1.5pt}
  \def\arraystretch{1.5}
  \begin{tabular}{ccc}
    \includefig{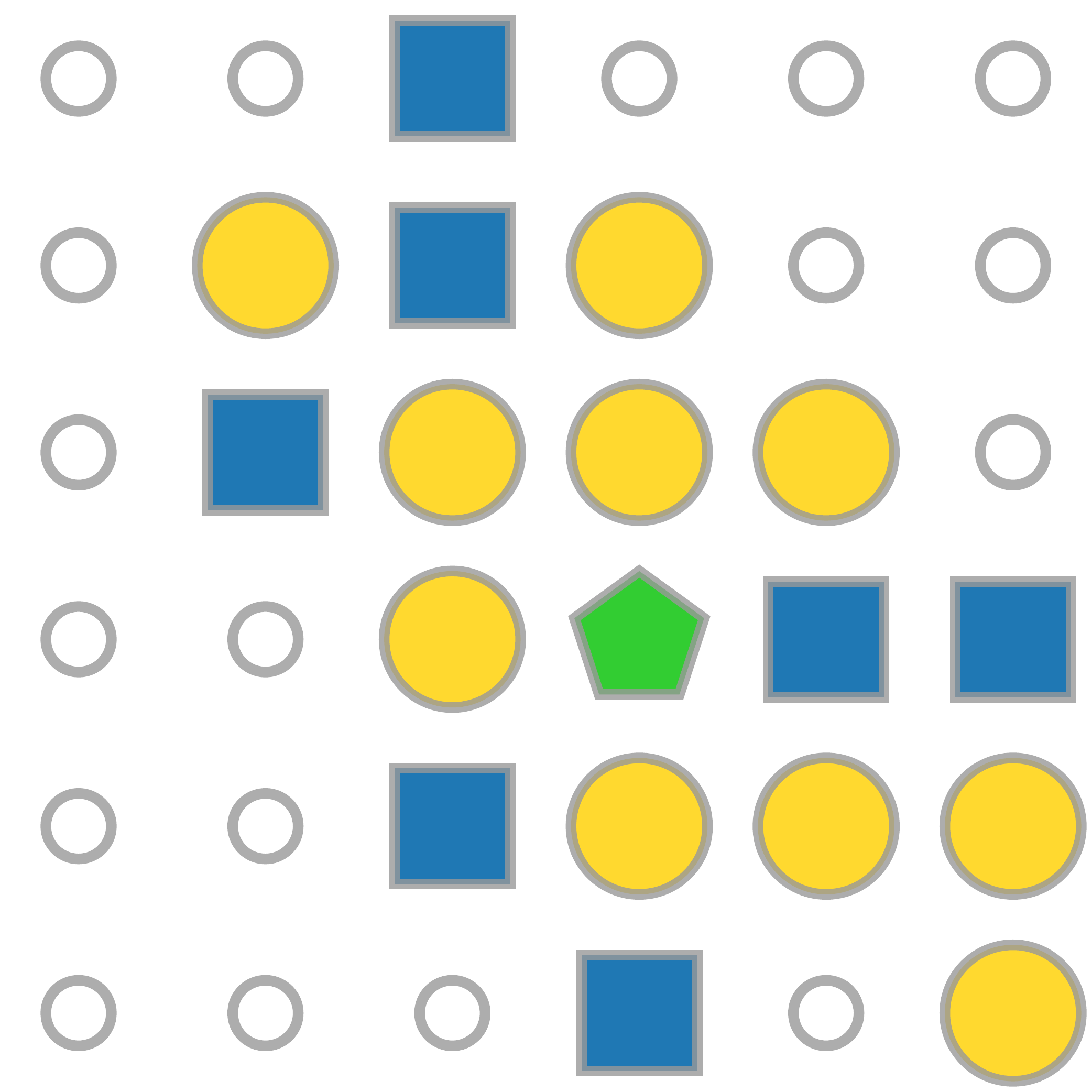} &  \includefig{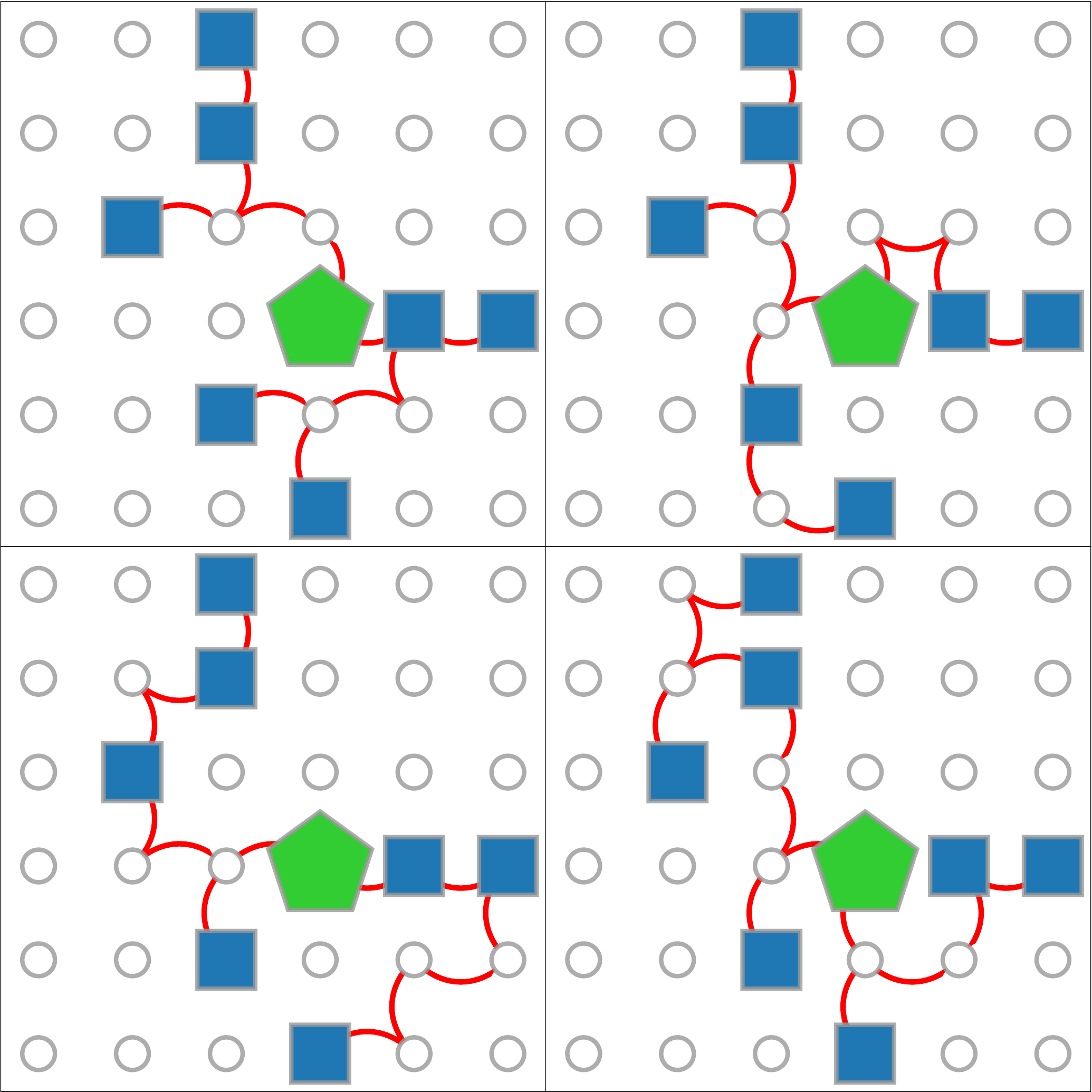} & \includefig{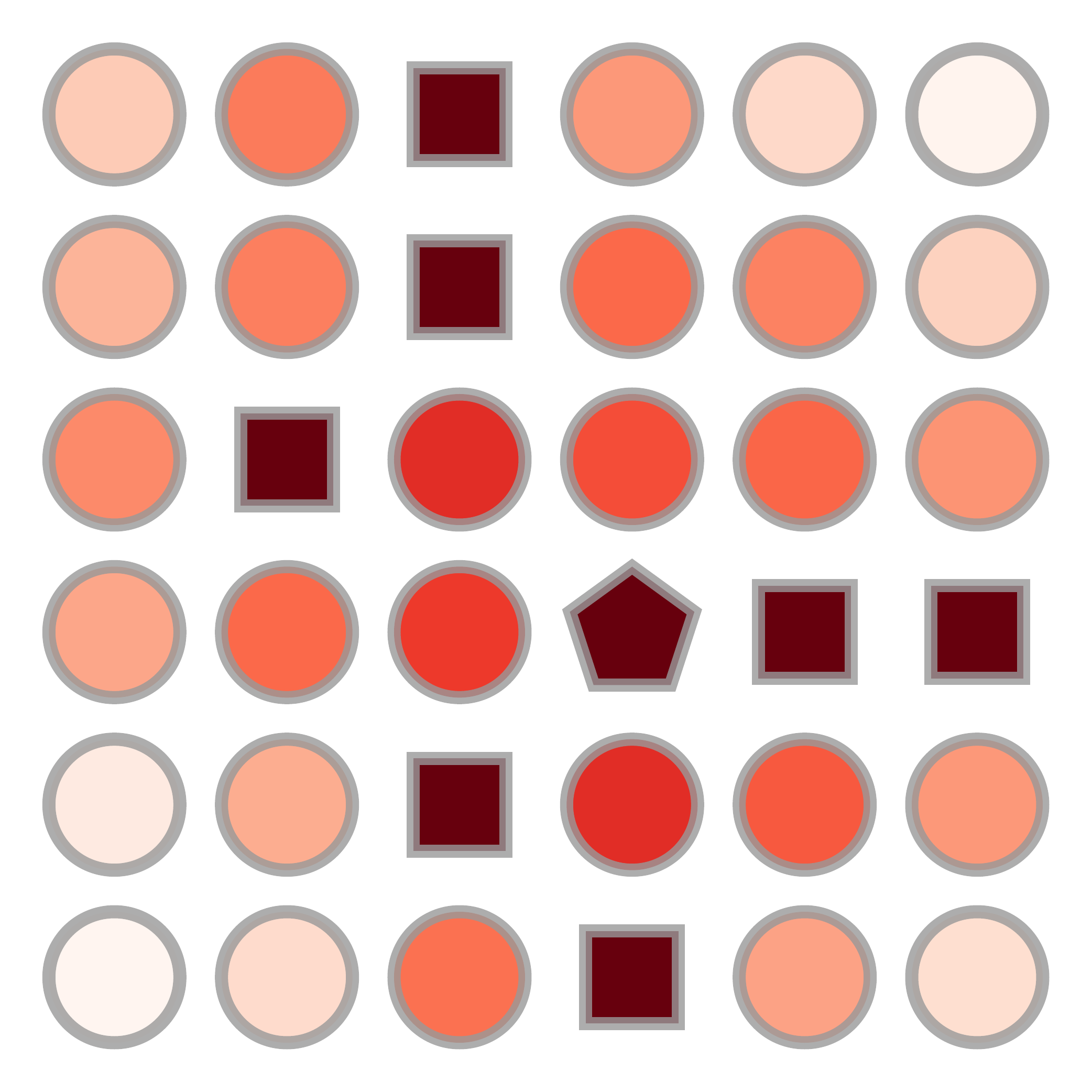} \\
    {\small $(a)$ Input} & {\small $(b)$ Tree sampling} &  {\small $(c)$ Inference} \\
    \multicolumn{3}{c}{\includegraphics[width=2\figwidth]{{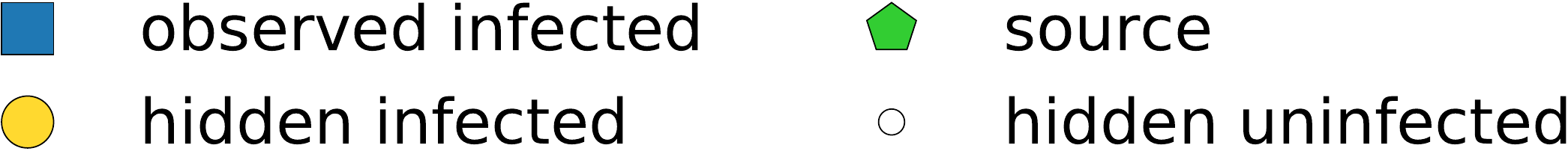}}} \\
    \multicolumn{3}{c}{\includegraphics[width=1.2\figwidth]{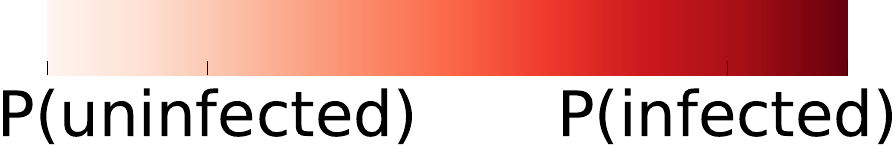}} 
  \end{tabular}  
  \caption{Overview of our approach:
  $(a)$ given input observed infections;
  $(b)$ Steiner trees are sampled, each representing a possible cascade explaining the observation;
  $(c)$ node infection probabilities are estimated from the samples;}
  \label{fig:query-illus}
\end{figure}


An illustration of our setting and approach is shown in Figure~\ref{fig:query-illus}.
To simplify visualization, the example is shown on a 2-$d$ grid, 
however, the idea is general for any network. 
Figure~\ref{fig:query-illus}$(a)$ shows the ground-truth cascade.
Among the infected nodes some are observed (blue squares~\obsnode) while most are unknown (yellow circles~\hiddeninfnode). 
The uninfected nodes (small white circles~\uninfnode) are also unknown.
Figure~\ref{fig:query-illus}$(b)$ shows four different cascades sampled from the grid. 
The sampled cascades are modeled as Steiner trees spanning a set of terminals --- 
the observed infected nodes. 
Using the sampled cascades we estimate the infection probabilities 
for all the nodes whose infection status is unknown ---
shown in Figure~\ref{fig:query-illus}$(c)$.

Our main contribution are as follows:
\begin{itemize}
\item We define the probabilistic cascade-reconstruction problem (Section~\ref{sec:problem}), which makes weaker assumptions compared to methods such as NetFill~\cite{sundareisan2015hidden}, thus offering more robustness. 
\item To solve the cascade-reconstruction problem, we study the problem of sampling Steiner trees with a given set of terminals, and propose two algorithms with provable guarantees on the sampling distribution (Section~\ref{section:sampling}). 
\item We empirically demonstrate that our sampling-based cascade-reconstruction approach outperforms other baselines in terms of both node and edge recovery (Section~\ref{section:experiments}). 
\end{itemize}

%

%% file: related.tex
\section{Related work}

\noindent
{\bf Cascade reconstruction.}
Diffusion processes have been studied in the literature extensively under different viewpoints. 
Nevertheless, the problem of ``reverse engineering'' an epidemic 
so as to reconstruct the most likely cascade, 
and the most likely infected nodes, or propagation edges, has received relatively little attention.
Lappas et al.~\cite{lappas2010finding} study the problem of 
identifying $k$ seed nodes, or \textit{effectors}, 
in a partially-activated network, 
which is assumed to be in steady-state under the independent-cascade model ({\ic}). 
The problem has also been considered in different settings. 
For example, recent works by Rozenshtein et al.~\cite{rozenshtein2016reconstructing} 
and Xiao et al.~\cite{xiao2018reconstructing}
consider different cascade-reconstruction problems in a temporal setting, 
where interaction timestamps, or infection times, respectively, are assumed to be known.

A related line of research focuses on finding the source node (root) of a given cascade, 
or a number of cascades.
Shah and Zaman~\cite{shah:11:culprit} formalize the notion of \textit{rumor-centrality} 
to identify the single source node of an epidemic under the susceptible-infected model ({\si}), 
and provide an optimal algorithm for $d$-\textit{regular trees}.
The more recent study of Farajtabar et al.~\cite{farajtabar2015back} considers the problem of identifying a single source given multiple partially-observed cascades.
In our case, we focus on detecting the unobserved infected nodes. 

The work closest to ours is the work by Sundareisan et al.~\cite{sundareisan2015hidden}.
Similar to our problem formulation, they aim at recovering hidden infections under a non-temporal setting where the source nodes are unknown.
However, they make two key assumptions. First, their method NetFill is tailored specifically for the \si model.
Second, they assume that the fraction of observed infections is known.
In contrast, our method does make such assumptions, making it more robust. 
On the other hand, their method also attempts to recover the infection sources while our method relies on external source-inference methods. However, we empirically show that the integration of our method  with simple distance-based source-inference heuristics can still significantly outperform NetFill. Finally, while NetFill  makes binary decisions, our method produces infection probability estimates: tackling this problem probabilistically suits more the stochastic nature of a variety of diffusion processes. 


\smallskip
\noindent
{\bf Sampling trees.} 
Sampling combinatorial structures has long been an active research area.
In particular, 
sampling random spanning trees~\cite{mosbah1999non, broder1989generating, wilson1996generating, propp1998get} is an extensively-studied problem. 
On the other hand, the problem of sampling Steiner trees has not received much attention in the literature.
Related work on Steiner trees has mainly considered the optimization aspect, such as minimum Steiner tree problem and its variants\cite{hauptmann2015compendium}.

%% file: problem.tex
\section{Problem definition}
\label{section:problem}

\noindent{\bf Infection cascade model.} We consider a reciprocal non-symmetric directed contact graph $G = (V,E,p)$. 
By reciprocal we mean that $(u,v) \in E$ implies $(v,u) \in E$. 
We assume that under an infection-propagation model, a dynamic process takes place in the network, and 
for each edge $(u,v)\in E$, $p_{uv}$ denotes the probability that $u$ transmits the infection to $v$. 
These infection probabilities are non-symmetric, meaning that in the general case $p_{uv} \not= p_{vu}$.
The infection process starts at an initial source node $s$, i.e., patient-zero, 
which could potentially be unknown. 
We denote by $T^*$ the resulting cascade 
and the corresponding set of infected nodes by $V[T^*]$. We assume each node can be infected by only one other node and infected nodes do not recover.\footnote{These assumptions are aligned with \si and \ic models while ruling out Linear Threshold~\cite{kempe2003maximizing} and \sir models.} Under these assumptions, the cascade $T^*$ forms a tree. 

\noindent{\bf Observation model.} For each node $u \in V$ we use the variable $y_u \in \{0,1\}$  to denote its infection state such that $y_u = 1$ if $u \in V[T^*]$ and $0$ otherwise. We assume that we can only partially observe the infection state of the network and define the set of observations $\obs=\set{(u, y_u)}$ as the set of nodes with known infection states. We use $X$ and $\obsuninfected$ to refer to the nodes that are reported to be infected and uninfected, respectively. Thus, abusing the notation slightly, for all $u \in V \setminus \obs$, the value of $y_u$ is unknown. We use $\proba{y_u = 1 \mid \obs}$ to denote the probability that a node $u \in V \setminus \obs$ is infected when we are given the partial observation set $\obs$. 

\noindent{\bf Probabilistic cascade reconstruction.} Given a partially observed infection propagation process, our aim is to reconstruct the infection cascade by inferring the hidden infection state of the unreported nodes.  Thus, we consider the following problem.

\begin{problem}[Probabilistic cascade reconstruction]
\label{problem:prob-cascade-reconstruction}
Given a reciprocal non-symmetric directed contact graph $G = (V,E,p)$, 
and a set $\obs=\set{(u, y_u)}$ of partial observations, reconstruct the infection cascade by inferring the hidden infection state of the unreported nodes. 
\end{problem}

\noindent{\bf Complexity of Problem~\ref{problem:prob-cascade-reconstruction}.} 
We now show that, 
computing the probability $\proba{y_u = 1 \mid \obs}$  that a node $u \in V \setminus \obs$ is in infected state is a \sharpP-hard problem. First, notice that an infection can propagate only via the infected nodes in a graph. Hence, we can simply handle the observations on uninfected nodes by removing such nodes and their adjacent edges from the graph as these nodes can neither become infected nor transmit infection. Thus, for a given $u$, to compute $\proba{y_u = 1 \mid \obs}$ in $G = (V,E,p)$, we can simply assign $V = V \setminus U$ and $E = E \setminus \{(u,v) \cup (v,u)\}$, for all $u \in U$. We provide our complexity analysis based on this. 

\begin{lemma}
\label{eq:problem1Complexity}
Given a probabilistic graph $G=(V,E,p)$, 
a partially observed infection cascade $T^*$ with 
 source $s$, and 
the corresponding subset $X \subseteq V[T^*]$ of partially observed infected nodes, 
computing $\proba{y_u = 1 \mid \obs = X}$ for any $u \in V \setminus X$ is 
\sharpP-hard. 
\end{lemma}

\begin{proof}
We prove this by a reduction from the counting problem of $s$-$t$ connectedness in a directed graph~\cite{valiant1979complexity}. Given a directed graph $G = (V,E)$ and two nodes $s,t \in V$, $s$-$t$ connectedness problem is to count the number of subgraphs of $G$ in which there is a directed path from $s$ to $t$. We reduce this problem to the problem of computing marginal infection probabilities as follows. Let $\obs = \emptyset$ and let $T$ denote a random infection cascade with infection source $s$. For any random tree $T$,  let $\mathbbm{1}{[t \in V[T]]}$ denote the indicator random variable that equals $1$ if $t \in V[T]$ for a given node $t$. It is easy to see that 
\begin{eqnarray*}
\proba{y_t = 1 \mid \obs = \emptyset}
& = & \parExp{T \sim \mathcal{T}}{\mathbbm{1}{[(t \in V[T])]}} \\ 
& = & \sum_{T \sqsubseteq \mathcal{T}} \proba{T} \mathbbm{1}{[t \in V[T]]} \\ 
& = & \sum_{g \sqsubseteq G} \proba{g} \path_g(s, t),
\end{eqnarray*}
where $g$ is a possible world sampled from graph $G$ by keeping each edge $(u,v) \in E$ with probability $p_{uv}$ and $\path_g(s,t)$ is an indicator variable that equals $1$ if $t$ is reachable via a path from $s$ in $g$. Therefore, for computing $\proba{y_t = 1 \mid \obs = \emptyset}$, we solve the $s$-$t$ connectedness counting problem. It is shown in ~\cite{valiant1979complexity} that $s$-$t$ connectedness is \sharpP-complete, and thus computing $\proba{y_t = 1 \mid \obs = \emptyset}$ is \sharpP-hard. 
\end{proof}

Given the \sharpP-hardness of our problem, we resort to Monte-Carlo sampling 
for estimating $\proba{y_u = 1 \mid \obs = X}$, 
for all $u \in V \setminus X$. 
To do so, we generate a sample of probable cascades that span $X$ 
by sampling directed Steiner trees from the set of all directed Steiner trees on $G$ with terminals $X$,
with probability proportional to their likeliness to happen. 
Thus, the operational definition for the cascade-reconstruction problem 
we consider in this paper is the following.

\begin{problem}[Random Infection Tree]
\label{problem:random-infection}
Given a weighted directed contact network $G = (V,E,p)$, 
a root node $r \in V$, and a set $\terminals \subseteq \nodes$ of terminal nodes, 
sample a directed Steiner tree $\rsttree$, rooted at $r$,
with terminal nodes \terminals, and with probability proportional to $Pr(\rsttree) \propto \prod\limits_{\substack{(u, v) \in \rsttree}} p_{uv}.$
\label{prob:rooted}
\end{problem}


%

%% file: sampling.tex
\section{Random Steiner-tree sampling}
\label{section:sampling}

Given a symmetric directed network $G=(V,E,p)$, 
a set of terminal nodes $X \subseteq V$, 
and a root node $r \in V$, 
our aim is to sample a random tree $\rsttree$ from the graph $G$ containing the terminal nodes $X$.
Our sampling approach is based on a random walk defined on $G$.
Since $G = (V,E,p)$ is not necessarily stochastic 
we need to define the Markov chain $\chainG = (V, \chainE, w)$, 
over the graph $G$, 
for which we perform the random walk. 

The random walk for the sampler is performed on the {\em reverse} direction of the edges of~$G$.
In other words the edge set of the Markov chain $\chainG$ is defined as $\chainE = E^T$, 
that is, the transpose of the edges of $G$.\footnote{We simply define the chain from the transposed graph as we will be using anti-arborescence model of random-walk based tree samplers.} 

Let $\Nin(u)$ denote the set of in-neighbors  of a node $u$ in~$G$, 
and let $p(u) = \sum_{v \in \Nin(u)} p_{vu}$ denote the weighted in-degree of $u$.
We also denote by $\Nout(u)$ the set of out-neighbors  of $u$ in~$G$.
Furthermore, we denote by $\inNeighbor{u}$ and $\outNeighbor{u}$ 
the sets of in-neighbors and out-neighbors, respectively, in the Markov chain \chainG.
To ensure that the Markov chain $\chainG$ is stochastic 
we set 
$w(u,v) = p_{vu}/p(u)$
so that the weighted out-degree 
$w(u) = \sum_{v \in \Nin(u)} w(u,v) = \sum_{v \in \outNeighbor{u}} w(u,v)$
is equal to $1$ for all the nodes $u$ in $\chainG$.
Our tree sampling method performs a random walk on $\chainG$ according to edge weights $w(u,v)$.
%
%
%
Let 
\[
w(\rsttree) = \prod\limits_{\substack{(u, v) \in \rsttree}} w(u,v)
\]
denote the weight of a Steiner tree $\rsttree$. 
Assume that we are able to sample a Steiner tree $\rsttree$, from the Markov chain $\chainG$, 
with probability proportional to the weight of the tree $w(\rsttree)$.
Notice that after normalization of the weights, 
we have 
\begin{align}
\label{eq:weightNormBias}
\prod\limits_{\substack{(u, v) \in \rsttree}} p_{uv} = w(\rsttree) \prod_{u \in V[\rsttree] \setminus \{r\}} p(u).
\end{align}

Thus, to solve Problem~\ref{prob:rooted}, 
we need to be able to sample a Steiner tree $\rsttree$ from $\chainG$
with probability proportional to $w(\rsttree)$:
once this is achieved, we use \emph{sampling importance resampling}~\cite{tears}
(SIR) to correct for the difference. 
SIR aims at sampling $x$ from \textit{target distribution} $\pi(x)$ via two steps: 
1) sampling from proposal distribution $q(x)$ and
2) resampling the samples with probability proportional to $\frac{\pi(x)}{q(x)}$. 
In our case, we resample each Steiner tree returned by the proposed algorithms with probability proportional to $\left( \prod_{u \in V[\rsttree] \setminus \{r\}} p(u) \right)$ so as to achieve the target  distribution, and thus solving Problem~\ref{prob:rooted}. We address next the problem of sampling a Steiner tree with probability proportional to its weight. 

\begin{problem}[Random Steiner Tree]
\label{problem:random-tree}
Given a Markov chain $\chainG = (V, \chainE, w)$, 
a root node $r \in V$, and a set of terminal nodes $\terminals \subseteq V$,
sample a Steiner tree $\rsttree$, rooted at $r$,
with terminal nodes \terminals, and with probability proportional to 
\[    
w(\rsttree) = \prod\limits_{\substack{(u, v) \in \rsttree}} w(u,v).
\]
\label{prob:rootedChain}
\end{problem}

The difference between Problems~\ref{problem:random-infection} and~\ref{problem:random-tree}, 
is that Problem~\ref{problem:random-tree} assumes that the input graph is a Markov chain, 
i.e., the probabilities out of each node sum up to 1.
A general input probabilistic graph, can be turned to Markov chain 
using the weighting scheme transformation outlined above.

Notice that since $G$ is reciprocal (non-symmetric) directed, 
the graph $\chainG$ is strongly connected, hence, the chain is irreducible. 
However, it is not time-reversible due to the non-symmetricity of the transition probabilities. 
In the next sections, we propose two sampling algorithms that can produce 
non-uniformly distributed Steiner trees from irreducible but not time-reversible Markov chains. 

\subsection{TRIM algorithm}
\label{sec:cutSampling}
\input{cut}

\subsection{{Loop-erased random walk algorithm}}
\label{sec:lerwSampling}
\input{loop_erased}

\input{compare_algs}

%% file: cut.tex
Our first algorithm, \algcut, 
receives as input the edge weights of the Markov chain \chainG, 
a root $r\in V$ and a set of terminal nodes $X\subseteq V$.
\algcut first samples a random spanning tree $\rsptree$ with root $r$ from the chain 
\chainG and then trims $\rsptree$ into a Steiner tree $\rsttree$ rooted at $r$, 
such that all the leafs of $\rsttree$ are terminal nodes.  

Before we proceed to present the details of the algorithm \algcut,
we first discuss how to sample a random spanning tree for a given root $r$. 
This is a classic problem studied extensively in the literature~\cite{broder1989generating,mosbah1999non,wilson1996generating}. 
Note that in our case the chain $\chainG$ is not time-reversible. 
Hence, we cannot use Broder's algorithm~
\cite{broder1989generating} as it can only operate on time-reversible Markov chains 
defined on unweighted undirected graphs. 
Similarly, the algorithm by Mosbah and Saheb~\cite{mosbah1999non} is devised for reversible chains 
defined on weighted but undirected graphs hence is also not suitable for our setting. 

Thus, to generate a random tree~$\rsptree$ from a weighted directed graph~$\chainG$, 
with probability distribution different than the uniform, 
we use Wilson's Random\-Tree\-With\-Root algorithm~\cite{wilson1996generating}, 
which works on any irreducible chain that is not necessarily time-reversible. 

Now we give the details of  the \algcut algorithm. 
The pseudocode is provided in Algorithm~\ref{alg:cut}. 
The algorithm first samples a spanning tree $\rsptree$ with root $r$. 
Then, the sampled spanning tree $\rsptree$ is given as input to the {TrimTree} procedure, 
which removes all the subtrees that are below each terminal node $x \in X$, 
which then become {\em leafs} in the tree $\rsptree$.
The trimmed tree $\rsttree$ is returned as the output of \algcut. 
It is clear that all terminal nodes $X$ are contained in $\rsttree$, 
and also that the set of leafs of $\rsttree$ is a subset of $X$.

\begin{algorithm}[t!]
\caption{\algcut}
\label{alg:cut}
\Indm
{\small
\SetKwInOut{Input}{Input}
\SetKwInOut{Output}{Output}
\SetKwComment{tcp}{//}{}
\Input{Markov chain $\chainG = (V,\chainE,w)$, root $r$, set of terminals $\terminals$}
 \Output{Random Steiner tree $\rsttree$}
}
\Indp
\BlankLine
$\rsptree \gets \text{{RandomTreeWithRoot}(r)}$ \;
$\rsttree \gets \text{{TrimTree}}(\rsptree)$ \;
\end{algorithm}

Before we start analyzing the distribution of the sampled Steiner trees $\rsttree$ returned by \algcut, 
we first introduce some concepts used in our analysis.

We use $\rsptrees(\chainG)$ to denote the set of all directed spanning trees of $\chainG$ rooted at~$r$. 
Given $\rsttree$, we use $\rsptreest(\chainG) \subseteq \rsptrees(\chainG)$ 
to denote the set of spanning trees of $\chainG$, with root $r$, 
in which $\rsttree$ is contained as a subgraph. 

For any weighted directed chain $\chainG = (V, \chainE, w)$, with $V = \{v_1, \cdots, v_n\}$, we define its Laplacian matrix $\laplacian(\chainG)$ as follows: 
\[
  \laplacian_{ij}(\chainG) = 
  \begin{cases} 
   \sum_{v_i \in \inNeighbor{v_j}} w(v_i,v_j)& \text{if } i = j \\
   -w(i,j)       & \text{if } i \neq j \text{ and } (v_i,v_j) \in \chainE \\
   0 & \text{otherwise}
   
  \end{cases}
\]

Given $\chainG$ and a Steiner tree $\rsttree$, 
we denote the \emph{contracted graph} of $\rsttree$ on $\chainG$ as 
$\cgraph = (\cnodes, \cedges, w_c, \ell_c)$. 
The contraction of $\rsttree$ on $\chainG$ merges all the nodes of $\rsttree$ into a supernode $\supernode$. 
Thus, the node set of $\cgraph$ is given by $\cnodes = (V \setminus V[\rsttree]) \cup \{\supernode\}$. 
The contraction may create parallel edges from a node $u \in V \setminus V[\rsttree]$ to $\supernode$ 
if $\lvert \outNeighbor{u} \cap V[\rsttree] \rvert > 1$, 
where $\outNeighbor{u}$ denotes the out-neighbors of $u$ in $\chainG$. 
For each $u \in V \setminus V[\rsttree]$ such that $\lvert \outNeighbor{u} \cap V[\rsttree] \rvert > 1$, 
we merge the parallel edges into a single edge $(u,\supernode)$, 
set its weight $w_c(u, \supernode) = \sum_{v \in \outNeighbor{u} \cap V[\rsttree]} w(u,v)$, 
and set its label $\ell_c(u,\supernode) = \lvert \outNeighbor{u} \cap V[\rsttree]\rvert $. 
Similarly, for each $u \in V \setminus V[\rsttree]$ such that $\inNeighbor{u} \cap V[\rsttree] > 1$, 
we merge the parallel edges from $\supernode$ to $u$ into a single edge $(\supernode, u)$, 
set its weight $w_c(\supernode, u) = \sum_{v \in \inNeighbor{u} \cap V[\rsttree]} w(v,u)$, 
and set its label $\ell_c(\supernode, u) = \lvert \inNeighbor{u} \cap V[\rsttree] \rvert$. 
For the rest of the edges, we use their original weights from $\chainG$ and set their labels to $1$, 
i.e., for each $u,v \in V \setminus V[\rsttree]$, we set $w_c(u,v) = w(u,v)$ and $\ell_c(u,v) = 1$.

Let $\rsptrees(\cgraph)$ denote the spanning trees of $\cgraph$ with root $\supernode$. 
The following lemma depicts the relation between $\rsptrees(\cgraph)$ and $\rsptreest(\chainG)$. 

\begin{lemma}
\label{lemma:treeSurjection}
Given a weighted directed graph $\chainG = (V, \chainE, w)$, 
a Steiner tree $\rsttree$ of $\chainG$, 
and the contracted graph $\cgraph$ of $\rsttree$ on $\chainG$, 
define a structure-preserving mapping $f: \rsptreest(\chainG) \mapsto \rsptrees(\cgraph)$, 
where for a given tree $T \in \rsptreest(\chainG)$, 
a tree $f(T)$ is obtained by contracting $\rsttree$ on $T$.
The function $f$ is surjective and any spanning tree $T_c$ in $\rsptrees(\cgraph)$ is the image of  
$\prod_{(u,v) \in T_c} \ell_c(u,v)$ spanning trees in $\rsptreest(\chainG)$. 
Moreover
\[
\sum_{T \in f^{-1}(T_c)} \dfrac{w(T)}{w(T_X)} = w(T_c).
\]
\end{lemma}

\begin{proof}
To show that $f$ is surjective, we need to show that: 
(\emph{i}) each $\rsptree \in \rsptreest(\chainG)$ maps to a unique spanning tree $T_c \in \rsptrees(\cgraph)$ and 
(\emph{ii}) each spanning tree $T_c \in \rsptrees(\cgraph)$ is the image of at least one spanning tree in $\rsptreest(\chainG)$. 

Regarding (\emph{i}) it is clear that for each $\rsptree \in \rsptreest(\chainG)$, 
$f(\rsptree)$ is a unique spanning tree with root $\supernode$ in $\cgraph$. 
The existence and uniqueness hold by construction. 

Regarding (\emph{ii}) given a spanning tree $T_c \in \rsptrees(\cgraph)$, 
we define an ``expansion'' process, which is the reverse process of contracting $\rsttree$ on $\chainG$: 
we first replace the contracted node $\supernode$ in $T_c$ with $\rsttree$. 
For each $(u, \supernode) \in \cgraph$ such that $\ell_c(u, \supernode) = 1$, we can directly replace the edge $(u, \supernode) \in \cgraph$ with the corresponding original edge $(u ,v) \in \chainE$, where $\{v\} = \outNeighbor{u} \cap V[\rsttree]$. On the other hand, for each $(u, \supernode) \in \cgraph$ such that $\ell_c(u, \supernode) > 1$, we cannot simply insert the original edges in $\chainG$ that go from $u$ to the nodes of $T_X$ since we would then have $\ell_c(u, \supernode)$ paths from node $u$ to root $r$ hence the expansion process would not create a valid tree. 
Thus, for each $(u, \supernode)$ in $T_c$ with $\ell_c(u, \supernode) > 1$, 
the expansion process should perform $\ell_c(u, \supernode)$ different assignments, 
where in each assignment, a different $(u,v) : v \in \outNeighbor{u} \cap V[\rsttree]$ is selected to replace $(u, \supernode)$. 
Hence, for a given $T_c$, 
the number of feasible assignments that create valid spanning trees are given by $\prod_{(u,v) \in T_c} \ell_c(u,v)$. 

Finally, given the structure-preserving surjection and the way we assign the weights in $\cgraph$, 
it  follows that for any $T_c \in \rsptrees(\cgraph)$, we have 
\[
w(T_c) = \sum_{T \in f^{-1}(T_c)} \dfrac{w(T)}{w(T_X)}.
\] 
\end{proof}

\begin{theorem}
Algorithm~\ref{alg:cut} {\em (}\algcut{\em )} returns a random Steiner tree $\rsttree$ rooted at $r$ with probability proportional to $w(\rsttree) \cdot \det(\laplacian_r(\cgraph))$ where $\det(\laplacian_r(\cgraph))$ is the determinant of the Laplacian matrix of $\cgraph$ with row and column $r$ removed. 
\label{theorem:cut-proba}
\end{theorem}

\begin{proof}
First, notice that the algorithm
Random\-Tree\-With\-Root$(r)$ of Wilson~\cite{wilson1996generating}, 
returns a random spanning tree $\rsptree$ with probability proportional to $w(\rsptree)$. 
Then, using Lemma~\ref{lemma:treeSurjection} we have: 
\begin{align*}
\proba{\rsttree} &= \sum_{\rsptree \in \rsptreest(\chainG)} \proba{\rsptree} \\
                 &\propto \sum_{\rsptree \in \rsptreest(\chainG)} w(\rsptree) \\
                 &= \sum_{T_c \in \rsptrees(\cgraph)} \sum_{T \in f^-1(T_c)} w(T) \\ 
                 &= w(\rsttree) \cdot  \sum_{T_c \in \rsptrees(\cgraph)} w(T_c) \\
                 &= w(\rsttree) \cdot \det(\laplacian_r(\cgraph))
\end{align*}
where the last line follows from Tutte's Matrix Tree Theorem for weighted directed graphs~\cite{tutte1984graph}. 
\end{proof}

Theorem~\ref{theorem:cut-proba} implies that the \algcut algorithm 
returns a Steiner tree with probability proportional 
to its weight among all the Steiner trees that contain the same set of Steiner nodes: 
let $\rsttree$ and $T'_X$ be two Steiner trees returned by the \algcut algorithm, 
such that $V[\rsttree] = V[T'_X]$. 
Then, the contracted graph of $\rsttree$ on $\chainG$ and the contracted graph of $T'_X$ on $\chainG$ will be identical, 
which means that the sum of the weights of the spanning trees of $\cgraph$ with root $\supernode$, i.e., the term $\sum_{T_c \in \rsptrees(\cgraph)} w(T_c)$, will be the same for both $T_X$ and $T'_X$. 
On the other hand, our aim is to sample a Steiner tree with probability proportional to its weight among all the Steiner trees 
(that could contain different set of Steiner nodes). 
Thus, to remove the bias introduced by the term $\sum_{T_c \in \rsptrees(\cgraph)} w(T_c)$, 
we employ importance resampling~\cite{tears} 
and resample each tree returned by the \algcut algorithm 
with probability proportional to $\left(\sum_{T_c \in \rsptrees(\cgraph)} w(T_c)\right)$.

\spara{Special case: unweighted graphs.} Consider the special case when $\chainG$ is unweighted, i.e.,  we perform a simple random walk on $\chainG = (V, \chainE)$. Then, it follows from Tutte's Matrix Tree Theorem for directed graphs~\cite{tutte1984graph} that $\lvert \rsptrees(\chainG) \rvert = \det(\laplacian_r(\chainG))$ where the Laplacian of unweighted directed $\chainG$ is given by: 
\[
  \laplacian_{ij}(\chainG) = 
  \begin{cases} 
   \lvert \inNeighbor{v_j} \rvert & \text{if } i = j \\
   -1      & \text{if } i \neq j \text{ and } (v_i,v_j) \in \chainE \\
   0 & \text{otherwise}
   
  \end{cases}
\]
Then, the \algcut algorithm returns a Steiner tree $\rsttree$ rooted at $r$ with probability equal to
\begin{align*}
\proba{\rsttree} &= \sum_{\rsptree \in \rsptreest(\chainG)} \proba{\rsptree} \\
                    &= \dfrac{\lvert \rsptreest(\chainG) \rvert}{\vert \rsptrees(\chainG) \rvert} \\
                    &= \dfrac{\sum_{T_c \in \rsptrees(\cgraph)} \prod_{(u,v) \in T_c} \ell_c(u,v)}{\det(\laplacian_r(\chainG))}
\end{align*}
where the last line follows from Lemma~\ref{lemma:treeSurjection}. 

%% file: loop_erased.tex
In this section, we propose a more efficient algorithm, based on loop-erased random walk (\alglerw), 
which is a more direct adaptation of the Random\-Tree\-With\-Root algorithm~\cite{wilson1996generating}.
The \alglerw algorithm can directly sample a Steiner tree $\rsttree$ 
from the chain $\chainG$ with probability proportional to its weight $w(\rsttree)$. 

The pseudocode for algorithm \alglerw is provided in Algorithm~\ref{alg:lerw}. 
\alglerw maintains a \emph{current tree} $\rsttree$, which initially consists of just the root $r$. 
As long as there are terminal nodes that are not in the tree, 
\alglerw performs a random walk in $\chainG$ starting from such $v \in X \setminus V[\rsttree]$, 
erasing loops (i.e., cycles) as they are created, 
until the walk encounters a node that is in the current tree $T_X$. 
This cycle-erased trajectory is then added to the current tree. 
The algorithm terminates once all the terminal nodes are added to the current tree $\rsttree$. 

\IncMargin{1em}
\begin{algorithm}[t!]
\Indm
{\small
\SetKwInOut{Input}{Input}
\SetKwInOut{Output}{Output}
\SetKwComment{tcp}{//}{}
\Input{Markov chain $\chainG = (V,\chainE,w)$, root $r$, terminals $\terminals$}
 \Output{Random Steiner tree $\rsttree$}
}
\Indp
\BlankLine
\For{$i \leftarrow 1 \ldots n$}{
  $\nextnode[i] \leftarrow false$\;
}
$\nextnode[\treeroot] \leftarrow nil$\;
$\intree[\treeroot] \leftarrow true$\;
\For{$i \in \terminals$}{
  $u \leftarrow i$\;
  \While{not $\intree[u]$}{
    $\nextnode[u] \leftarrow \mathrm{RandomSuccessor}(u)$\;
    $u \leftarrow \nextnode[u]$\;
  }
  $u \leftarrow i$\;
  \While{not $\intree[u]$}{
    $\intree[u] \leftarrow true$ \Comment*[r]{erase the loop}\
    $u \leftarrow \nextnode[u]$\;
  }
  \Return{$\nextnode$}\;
}
\caption{\alglerw}
\label{alg:lerw}
\end{algorithm}
\DecMargin{1em}

\begin{theorem}
The \alglerw algorithm 
returns a random Steiner tree $\rsttree$ rooted at~$r$ with probability proportional to $w(\rsttree)$. 
\label{theorem:lerw}
\end{theorem}


To prove this result, we will use an adaptation of the cycle-popping algorithm, 
devised by Wilson~\cite{wilson1996generating}, 
which was shown to be equivalent to loop-erased random walk when sampling a spanning tree. 

Given $\chainG$ and root $r$, 
following Wilson~\cite{wilson1996generating}, 
we associate to each non-root node $u$ an \emph{infinite} stack 
\[
S_{u} = [S_{u,1}, S_{u,2}, \cdots],
\] 
whose elements are states of the chain $\chainG$ such that 
\[
\proba{S_{u,i} = v} = \proba{\mathrm{RandomSuccessor}(u) = v},
\] 
for all $u$ and $i$ and such that all the items $S_{u,i}$ are jointly independent of one another.  We refer to the left-most element $S_u$ as the element that is at the \emph{top} of stack $S_u$, and by popping the stack $S_u$ we mean removing the top element from $S_u$. At any moment, the top elements of the stacks define a directed stack graph $\graph_\treeroot$ such that $\graph_\treeroot$ contains an edge $(u,v)$ only if the top element of $S_u$ is $v$. 
If there is a directed cycle $C$ in $\graph_\treeroot$, we pop $S_u$ for every $u \in C$, 
which creates another stack graph with cycle $C$ removed.

\begin{algorithm}[t]
\Indm
{\small
\SetKwInOut{Input}{Input}
\SetKwInOut{Output}{Output}
\SetKwComment{tcp}{//}{}
\Input{Markov chain $\chainG = (V,\chainE,w)$, root $r$, terminals $\terminals$}
 \Output{Random Steiner tree $\rsttree$}
}
\Indp
\BlankLine
  \While{$\graph_\treeroot$ has a cycle involing any item in $\terminals$}{
    pop any such cycle off the stack\;
  }
  $\rsttree \leftarrow $ the edges that are traversed on $\graph_\treeroot$ starting from each terminal to the root\;
  \Return{$\rsttree$}\;
  \caption{$\algcp$ procedure}
  \label{alg:cycle-popping}
\end{algorithm}

Wilson showed that when the cycle-popping process terminates, 
the stack graph $\graph_\treeroot$ is a directed spanning tree rooted at $r$~\cite{wilson1996generating}. 
Our aim instead is to obtain a Steiner tree: 
this translates to an early termination of the cycle-popping procedure 
as soon as the cycles do not contain any terminal node, 
thus, defining a Steiner tree on top of the stacks. 
The pseudocode of the cycle-popping procedure 
for sampling a Steiner tree is provided in Algorithm~\ref{alg:cycle-popping}. 

\begin{lemma}
$\algcp$ terminates with probability $1$ 
and the stack graph $\graph_\treeroot$ upon termination is a Steiner tree. 
\label{lemma:st}
\end{lemma}

\begin{proof}
We first show that $\graph_\treeroot$ is a Steiner tree. 
Suppose some terminal node $u$ is connected to another connected subgraph 
$\graph_\treeroot'$ induced by the top entries of the stacks. 
Since $\graph_\treeroot'$ does not touch the current tree, 
root $r$ does not belong in $\graph_\treeroot'$. 
Thus, it must have $|\graph_\treeroot'|$ edges, 
which implies that it contains a cycle, 
which is a contradiction since cycle popping terminates when no $u \in X$ is involved in a cycle. 

We now show that $\algcp$ terminates with probability $1$. 
Following Wilson~\cite{wilson1996generating}, 
when the cycle-popping procedure is performed until no node in $V$ is involved in a cycle, 
i.e., without early termination, 
it returns a spanning tree with probability $1$. 
This directly implies that it must also return a Steiner tree 
in a finite number of steps since every Steiner tree of $\chainG$ 
is a subgraph of at least one spanning tree of $\chainG$. 
\end{proof}

Using similar arguments as by Propp and Wilson~\cite{propp1998get}, 
we now show that the cycle-popping procedure with early termination 
is equivalent to loop-erased random walk that operates on an arbitrary ordering of the terminal nodes. 
First, notice that, instead of generating infinitely long stacks and looking for cycles to pop, 
$\alglerw$ uses the principle of deferred decisions and generates stack elements when necessary. 
$\alglerw$ starts with an arbitrary terminal node $u \in X$, 
does a walk on the stacks so that the next node is given by the top of the current node's stack. 
If $\alglerw$ encounters a loop, 
then it has found a cycle in the stack graph induced by the stacks that $\alglerw$ generates. 
Erasing the loop is equivalent to popping this cycle. 
When the current tree (initially just the root $r$) is encountered, 
the cycle erased trajectory starting from node $u$ is added to the tree. 
This process is then repeated for the other terminals that are not yet added to the current tree.

Given the equivalence between the methods $\alglerw$ and $\algcp$, 
we now provide the proof of Theorem~\ref{theorem:lerw} by reasoning with the cycle-popping procedure. 

\begin{proof} (of Theorem~\ref{theorem:lerw})
It follows from Lemma~\ref{lemma:st} that when $\algcp$ terminates, the stack graph $\graph_\treeroot$ induced by the stack entries containing the root node $r$ is a Steiner tree. 
Now, given a collection of cycles $C$, let $\mathbf{C}$ denote the event that $C$ is the collection of cycles that $\algcp$ pops before it terminates. Similarly, let $\mathbf{T_X}$ denote the event that a fixed Steiner tree $T_X$ is produced by the algorithm. Finally, let $\mathbf{C} \wedge \mathbf{T_X}$ denote the event that the algorithm popped the cycles $C$ and terminated with $T_X$. Following Propp and Wilson~\cite{propp1998get}, 
the order of cycles being popped until termination and the tree returned upon termination are independent, 
hence, we have 
\begin{align*}
\proba{\mathbf{T_X}} &= \sum_{C} \proba{\mathbf{C} \wedge \mathbf{T_X}} \\
&= \proba{\mathbf{T_X}} \sum_{C} P(\mathbf{C}) \\
&\propto w(T_X),
\end{align*}
where the last line follows from the fact that $\proba{\mathbf{T_X}}$ 
is the probability that the edges traversed on $\graph_\treeroot$  
starting from each terminal to the root is equal to $\rsttree$, 
which is given by the stack entries containing root $r$ upon termination. 
\end{proof} 


%% file: compare_algs.tex
\spara{Comparison.} Both \algcut and \alglerw rely on loop-erased random walk and
 run in time $\bigO(\tau)$, where $\tau$ is the \textit{mean hitting time} \cite{wilson1996generating}
 However, while \algcut, needs to visit all the nodes at least once, \alglerw is equipped with an early termination procedure, hence can terminate without having to visit all the nodes, making it more efficient than  \algcut in practice. Thus, we use \alglerw in our experiments.

%
%
 


%% file: experiment.tex
\section{Experimental evaluation}
\label{section:experiments}

In this section, we conduct an extensive empirical evaluation of the proposed approach and several baselines.

\subsection{Experimental setup}
\spara{Datasets.} We experiment on real-world graphs with both synthetic and real-world cascades. 
We use public benchmark graph datasets, in particular, 
\infectious,\footnote{http://konect.uni-koblenz.de/networks/sociopatterns-infectious}
\email,\footnote{http://networkrepository.com/ia-email-univ.php}
\student,\footnote{http://networkrepository.com/ia-fb-messages.php} 
\grqc,\footnote{http://snap.stanford.edu/data/ca-GrQc.html} and
\digg.\footnote{https://www.isi.edu/\textasciitilde{}lerman/downloads/digg2009.html}
The characteristics of these graphs are shown in Table~\ref{tbl:dataset}.
In addition, we perform a case study on a \lattice graph.



\begin{table}[t]
  \caption{Datasets}
  \label{tbl:dataset}
  \centering
  \begin{tabular}{lrrrr}
    \toprule
    Name         &   $|V|$   &    $|E|$     &  Assortativity \\  \midrule
    \infectious &     410   &     2\,765   &       0.0121   \\
    \email      &    1\,133 &     5\,451   &      -0.0007   \\
    \student    &    1\,266 &     6\,451   &      -0.0039   \\
    \grqc       &    4158   &    13\,428   &       0.1641   \\
    \digg       &  279\,631 &  1\,548\,131 &       0.0015   \\
    \bottomrule
  \end{tabular}
\end{table}

For all the datasets, except \digg, 
we generate synthetic cascades using the following diffusion models:
($i$) \si model with the infection probability $\beta$ set to 0.1;
($ii$) \ic model, where the infection probability on each edge is independently drawn from the uniform distribution $[0,1]$. 
For each graph that is originally undirected, we create directed copies of each undirected edge. 

For the \digg dataset, we experiment on real-world available cascades, which correspond to stories that propagate in the network.  In most of the cases, the activated nodes do not form a  connected component, hence, we extract the largest connected component as the cascade. We experiment using the top-10 largest cascades, with an average cascade size of $1868$.


\spara{Reconstruction methods.} We test and compare the following reconstruction methods.\footnote{Our implementation is available at \href{https://github.com/xiaohan2012/cascade-reconstruction-by-tree-samples}{github/xiaohan2012}}

\begin{squishlist}
\item {
    \ourmethod: our sampling-based method that returns, for each hidden node, the marginal probability of being in infected state. 
    We use \alglerw for the sampling of Steiner-trees and set the size of the Steiner-tree sample to $1000$.\footnote{We observed that increasing the sample size beyond $1000$ gives little marginal gain in accuracy.}
For the case when the real infection source is unknown, we use different root selection strategies to be used as a proxy to the real infection source. Obviously, the accuracy of our sampling approach is affected by the choice of root node. We discuss the different root selection strategies employed in the next sub-section.}
\item {
    \netfill: a method designed specifically for the \si model by Sundareisan et al.~\cite{sundareisan2015hidden}, which assumes a single propagation probability $\beta$ common to all edges. 
    This method additionally assumes that the fraction of observed infections are known.
    For completeness, we also test \netfill using \ic cascades where we set $\beta = 0.5$, 
    corresponding to the mean of the infection probabilities drawn from the uniform $[0, 1]$ distribution.}
\item {
    \pagerank: a baseline that ranks the nodes based on their Personalized PageRank scores~\cite{bahmani2010fast}, with the personalization vector initialized to $1/\abs{\obs}$ for the observed nodes and $0$ for the unobserved nodes. }
\item {
    \mintree: a baseline that constructs the minimum-weight Steiner tree where the weight of each edge $(u,v)$ is defined as $-\log p_{uv}$. Note that, this baseline is inspired by the approach in \cite{lappas2010finding} for constructing a Steiner tree with the minimum negative log-likelihood. Notice that, this baseline can also be considered as the \emph{time-agnostic} version of the minimum Steiner tree based reconstruction approach proposed in \cite{xiao2018reconstructing}. }
\end{squishlist}

\spara{Root selection.} we experiment with the following root-selection methods: 
\begin{squishlist}
\item {\mindistroot selects the ``centroid'' node with minimum weighted shortest path distance to the terminals, 
where the weight on each edge $(u,v)$ is defined as $-\log p_{uv}$.} 
\item {\pagerankroot chooses the node with the highest Personalized PageRank score as the root, where the personalization vector is initialized to $1/\abs{\obs}$ for observed nodes and to $0$ for unobserved nodes. }
\item {\trueroot is an oracle that simply returns the true infection source $s$. We consider this case to eliminate the effect of root selection on the final performance}
\end{squishlist}
We observe that, while \pagerankroot is more scalable than \mindistroot, \mindistroot gives slightly better performance in practice. 

\spara{Evaluation measure.} We evaluate the performance of the algorithms not only for recovering the nodes, but also recovering the edges: ideally a good method should discover many hidden infected nodes, along with the edges that the infection propagated. 

We use \emph{average precision} (\ap) to evaluate the quality of different reconstruction methods. 
The \ap measure summarizes a precision-recall curve as the weighted mean of precision achieved at each threshold, 
with the increase in recall from the previous threshold used as the weight. In particular, 
\[
\ap = \sum_n (R_n - R_{n-1}) P_n,
\]
where $P_n$ and $R_n$ is the precision and recall at the $n$-th thres\-hold, respectively.
The \ap measure is widely used for evaluating the accuracy of information retrieval (IR) tasks, where giving higher ranking to a small set of relevant documents against non-relevant ones is the primary goal. Notice that our task is very similar to an IR task, since in real-world scenarios, cascades tend to be small with respect to graph size, thus, a good reconstruction method should give high ranking to the actual hidden infected nodes (and the edges).  By convention, we exclude all observed nodes from evaluation. For each experiment setting, we report the results of experiments averaged over $100$ runs. 



\subsection{Case study on a lattice graph}

We perform a case study on simulated \si cascades on a 32$\times$32 lattice graph, 
where 20\% of the nodes are infected, among which 50\% are observed. Figure~\ref{fig:case-study} illustrates the behavior of different reconstruction algorithms by a color-coding of the predictions they produce. In contrast to \ourmethod that provides, for each hidden node, the probability to be infected, \netfill and \mintree can only produce binary predictions, which is reflected by their binary colored output. Notice that \ourmethod correctly assigns high probability to the actual hidden infections, which is shown by the assignment of dark red color to the actual infection region. We also observe that \netfill finds considerable amount of true positives \truepos, but also returns quite a few false positives \falsepos and false negatives \falseneg. Finally, \mintree, by design, produces small trees, thus, producing more false negatives \falseneg than the other two methods. 

%

\begin{figure*}[t]
  \centering  
    \newcommand{\includefig}[1]{\includegraphics[width=0.23\linewidth]{figs/case-study/#1.png}}
  \newcommand{\includelegend}[1]{\includegraphics[width=0.8\linewidth]{figs/case-study/#1}}
  \setlength{\figwidth}{4.0cm}
  \setlength{\tabcolsep}{3pt}
  \begin{tabular}{cccc}
    cascade with observation & \ourmethod & \netfill & \mintree \\
    \includefig{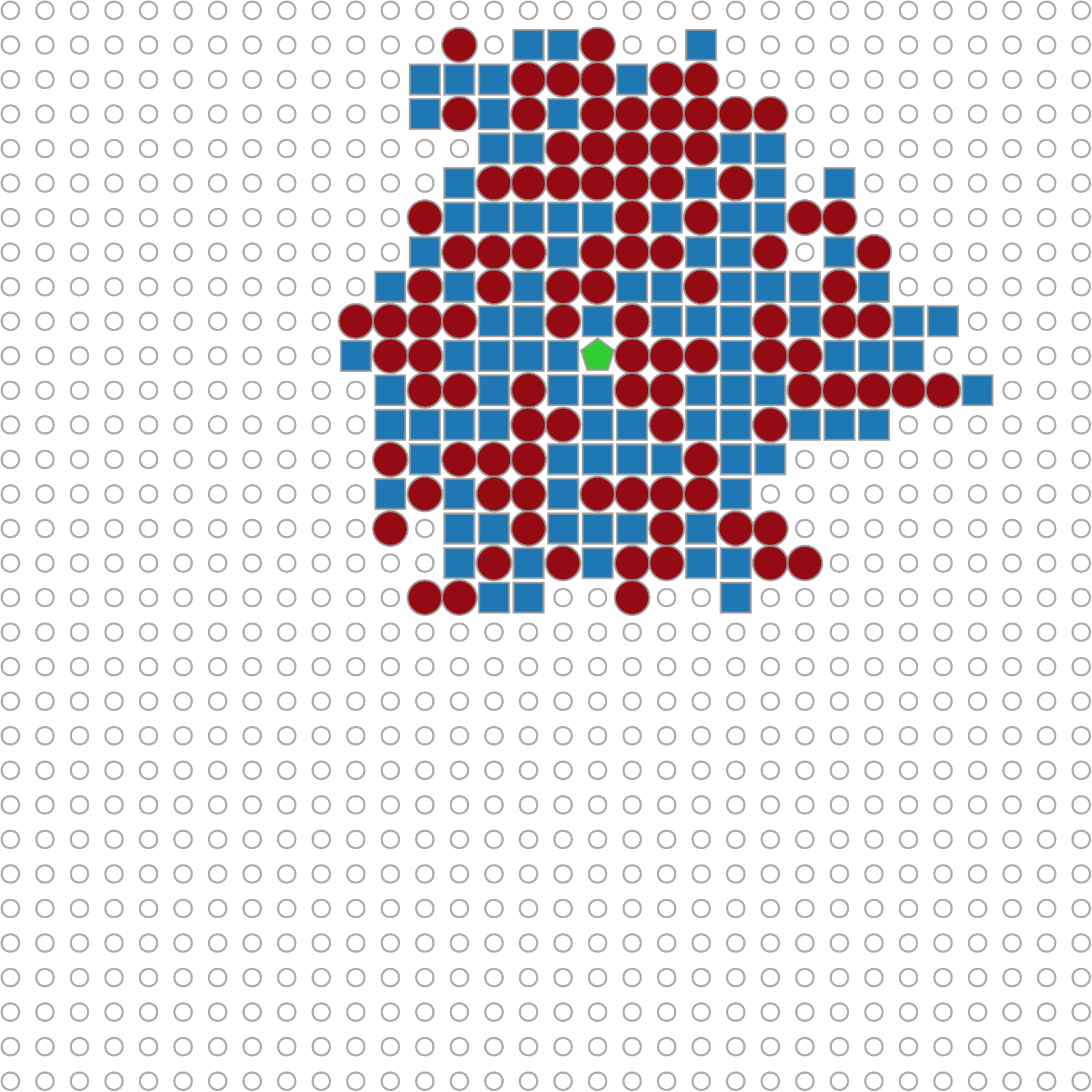} & \includefig{{our-min_dist}} & \includefig{{netfill}} & \includefig{{min-steiner-tree}}\\
    \multicolumn{2}{l}{\includegraphics[width=0.4\linewidth]{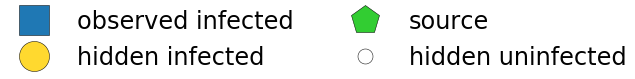}} & 
                                                                                                         \multicolumn{2}{c}{\includegraphics[width=0.4\linewidth]{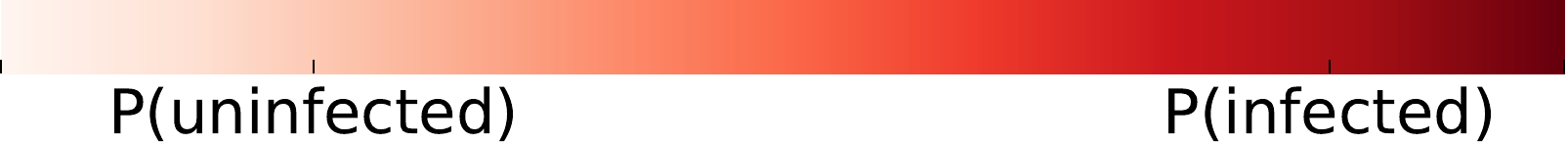}} \\
  \end{tabular}  
  \caption{
    Case study on \lattice graph. 
    Underlying cascade and observation is shown in left most figure.
    Output by different methods are shown in the remaining figures.
    Note that \ourmethod outputs node infection probability while \netfill and \mintree make binary decision (infected or not).
  }
  \label{fig:case-study}
\end{figure*}

\subsection{Node-level performance on synthetic cascades}
\label{subsec:node-level}

Next, we evaluate node-level \ap scores in different settings.
We experiment with various values of the following three parameters:
\begin{squishlist}
\item \emph{cascade generation model}: either \ic or \si;
\item \emph{observation fraction}: the fraction of the number of observed infections by the number of true infected nodes;
\item \emph{cascade fraction}: the fraction of the number of true infected nodes by the total number of nodes in the graph. 
\end{squishlist}

\spara{Effect of observation fraction.}
We fix cascade fraction to $0.1$ and vary the observation fraction from $0.1$ to $0.9$ at step size $0.1$.
Figure~\ref{fig:node-by-obs} shows node-level \ap scores under different graphs and cascade models. We see that \ourmethod (with both \mindistroot and \trueroot) and \pagerank
outperform \netfill\footnote{We are not able to run \netfill with observation fraction $<0.5$ using the implementation provided by the original authors.}
and \mintree by a large margin.
This result is consistent on all graphs we experiment and cascade models.
Notice that, though \pagerank is a simple heuristic, it is among the top-performing methods, 
even outperforming \ourmethod with \trueroot on \grqc by as much as $0.1$. 
We also see that, the effect of root selection plays an important role for most graphs.
For example, on all graphs except \grqc, with \trueroot, \ourmethod is the best,
while  with \mindistroot, \ourmethod falls behind \pagerank. 
Finally, as the observation fraction increases, all methods' performance tends to drop.
This is expected as more infections are observed, 
there are fewer hidden infections, which makes the task harder. 

\begin{figure*}[t]
  \centering  
  \newcommand{\includefig}[1]{\includegraphics[width=0.24\linewidth]{figs/cmp-baseline-obs-fractions/#1.pdf}}
  \newcommand{\includelegend}[1]{\includegraphics[width=0.8\linewidth]{figs/#1}}
  \setlength{\figwidth}{4.0cm}
  \setlength{\tabcolsep}{3pt}
  \begin{tabular}{cccc}
    \infectious  & \email & \student & \grqc \\
    \includefig{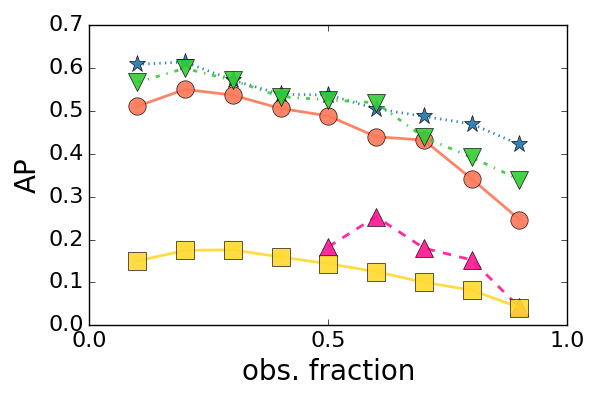} & \includefig{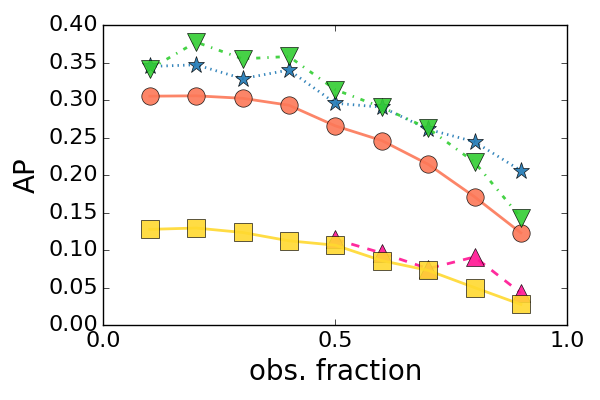} & \includefig{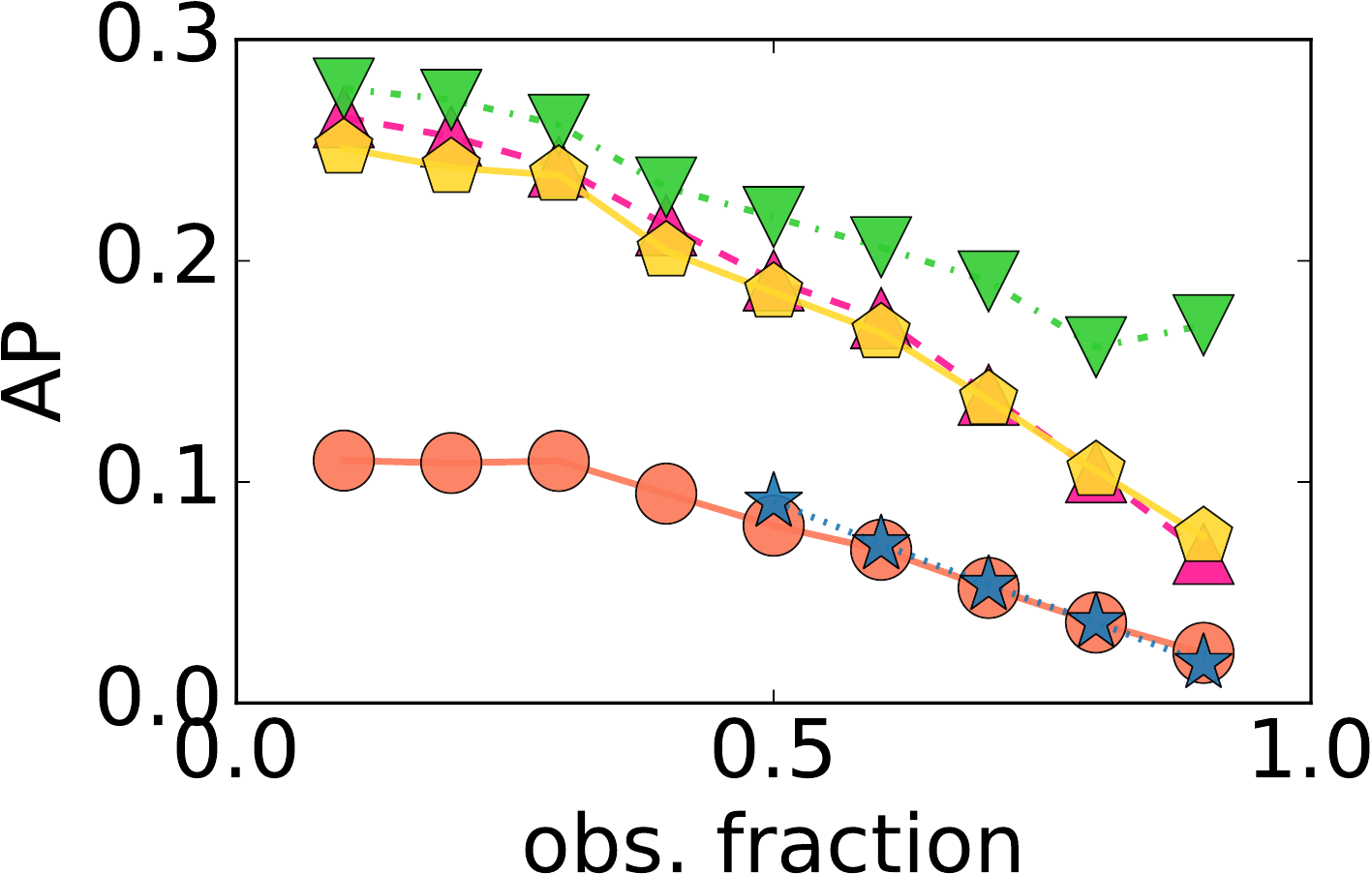} & \includefig{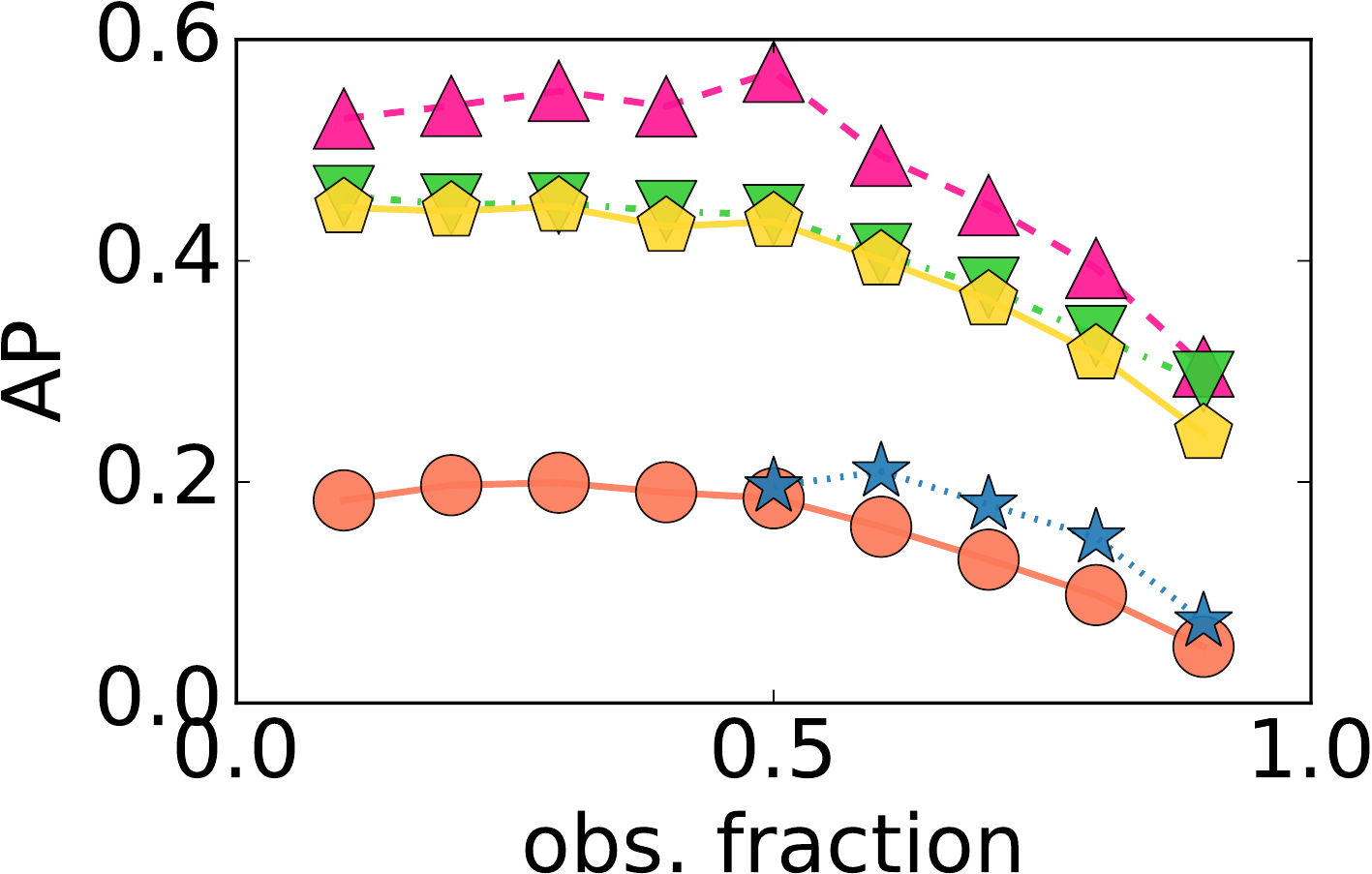} \\
    \includefig{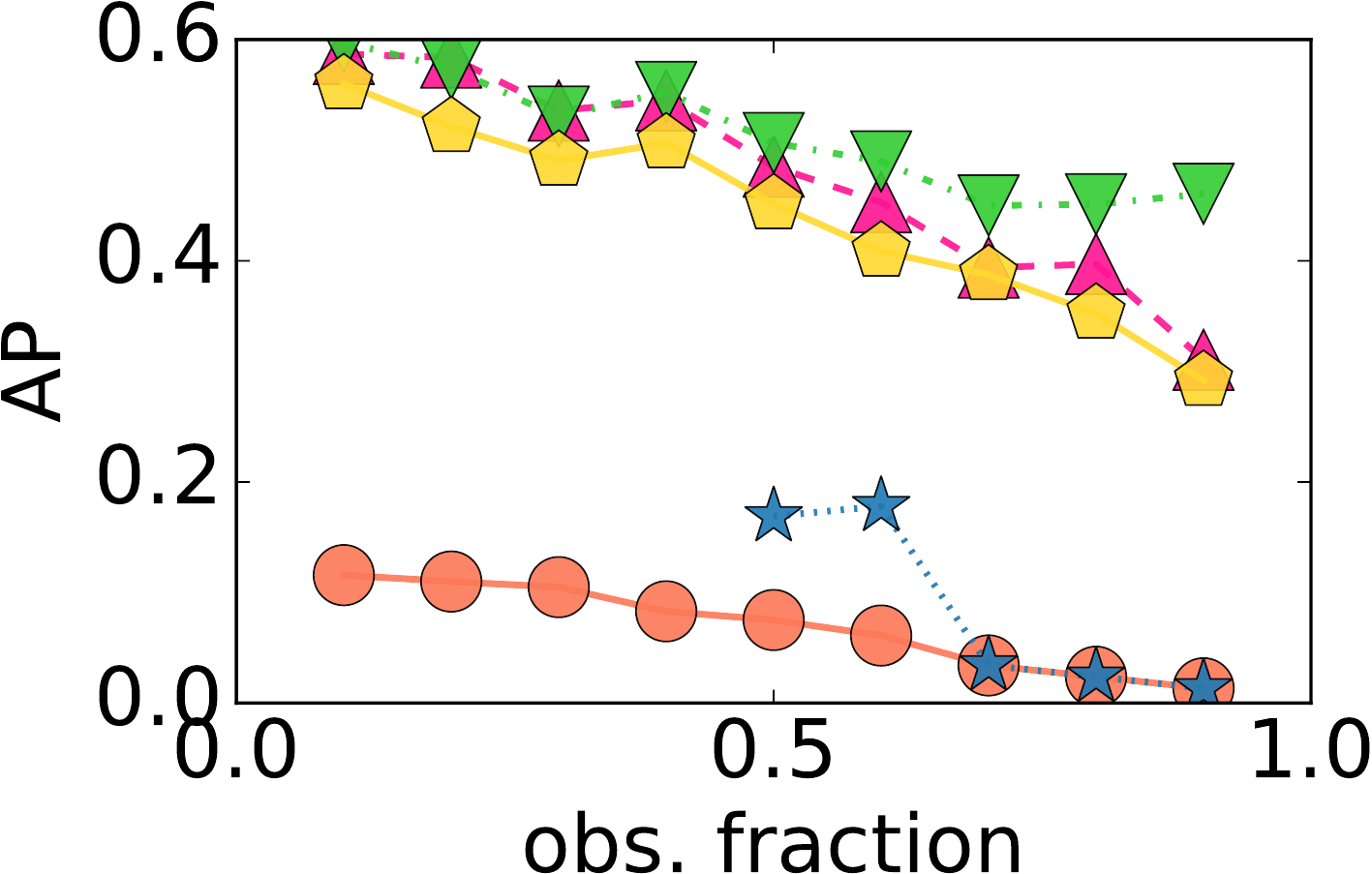} & \includefig{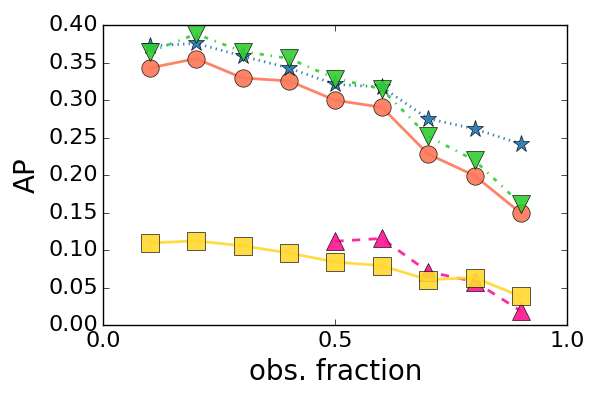} & \includefig{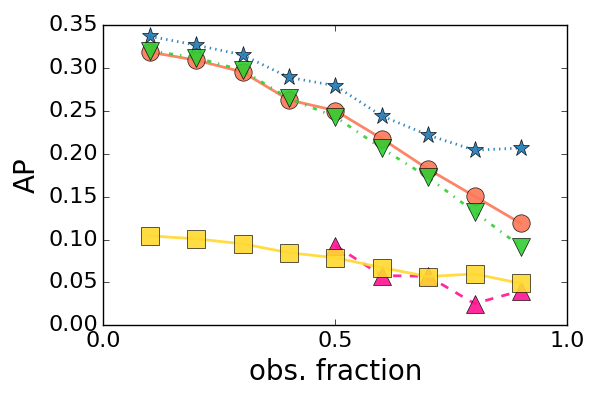} & \includefig{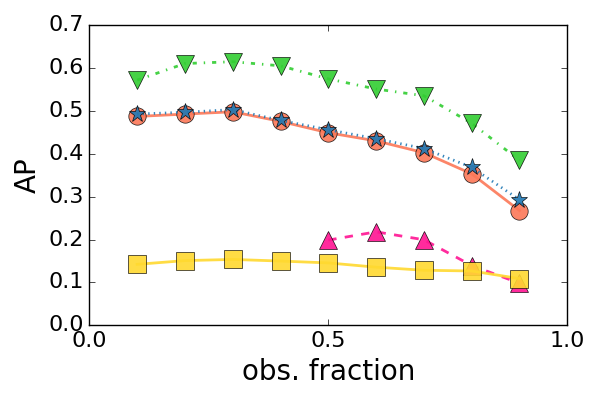} \\
    \multicolumn{4}{c}{\includelegend{{method_legend.pdf}}} 
  \end{tabular}  
  \caption{Node-level \ap score with respect to the fraction of observation: \textbf{top row} (\ic model), \textbf{bottom row} (\si model)}
  \label{fig:node-by-obs}
\end{figure*}

\spara{Effect of cascade fraction.}
We fix observation fraction to $0.5$ and vary the cascade fraction from $0.1$ to $0.5$ at step size $0.1$.
Figure~\ref{fig:node-by-cas} shows node \ap scores under different graphs and cascade models.
Again, \netfill and \mintree are out-performed by other methods.
For \ourmethod and \pagerank, \ap scores are very close except on \grqc, where \pagerank outperforms the others.
Also, as cascade fraction increases, all methods tend to give better \ap scores.

\begin{figure*}[t]
  \centering  
  \newcommand{\includefig}[1]{\includegraphics[width=0.24\linewidth]{figs/cmp-baseline-cascade-fractions/#1.pdf}}
  \newcommand{\includelegend}[1]{\includegraphics[width=0.8\linewidth]{figs/#1}}
  \setlength{\figwidth}{4.0cm}
  \setlength{\tabcolsep}{3pt}
  \begin{tabular}{cccc}
    \infectious  & \email & \student & \grqc \\
    \includefig{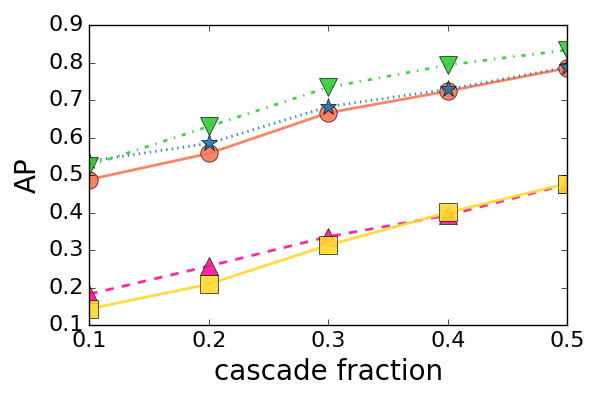} & \includefig{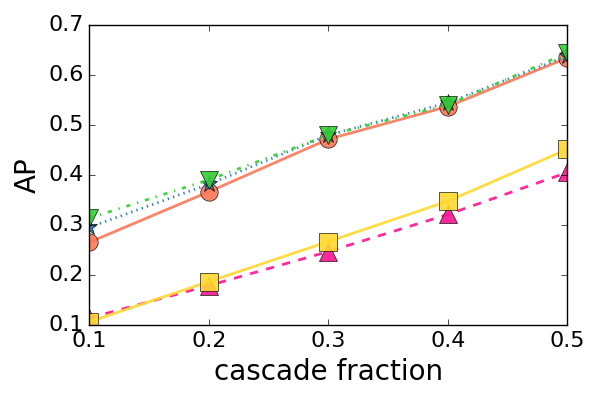} & \includefig{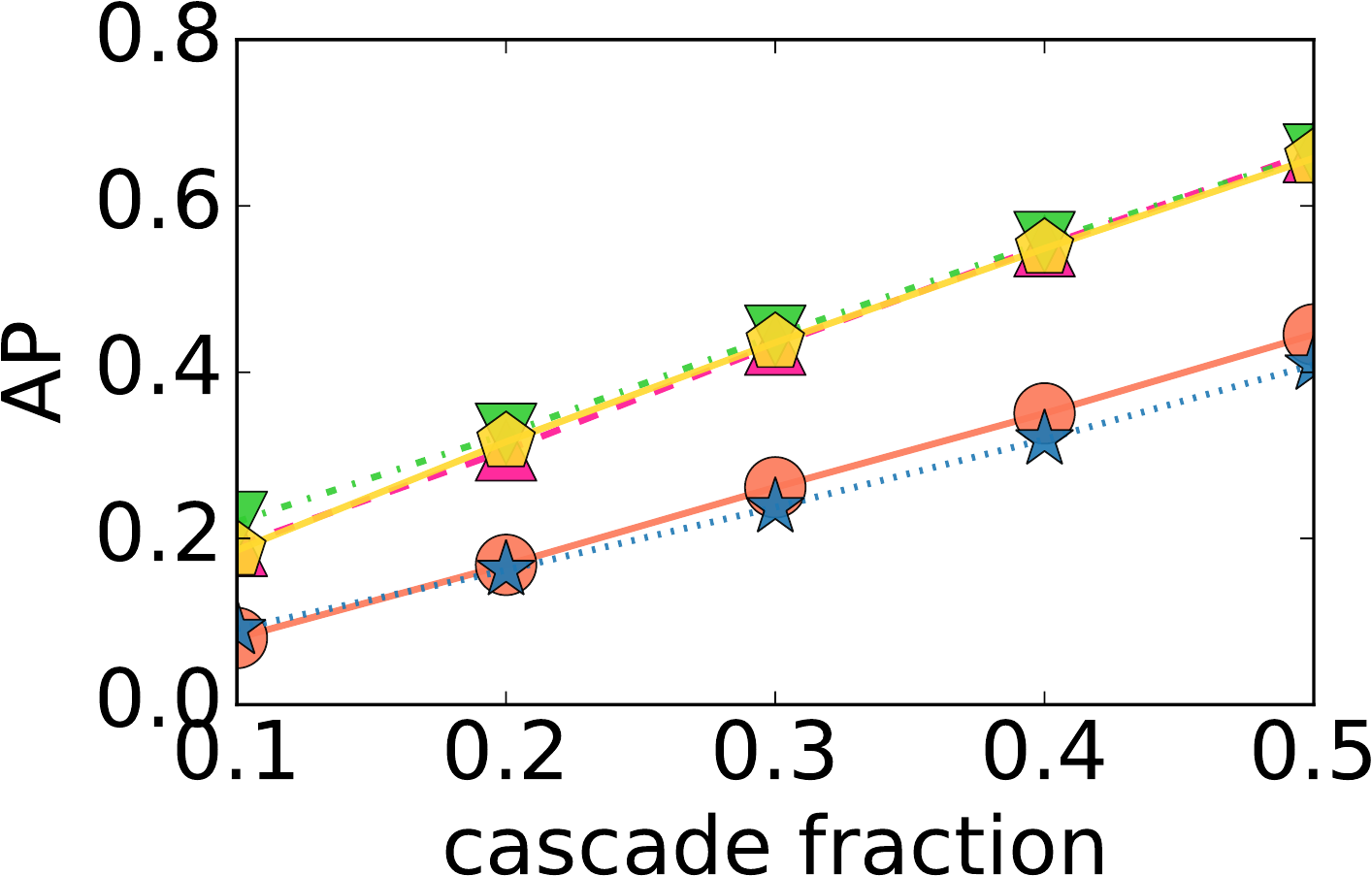} & \includefig{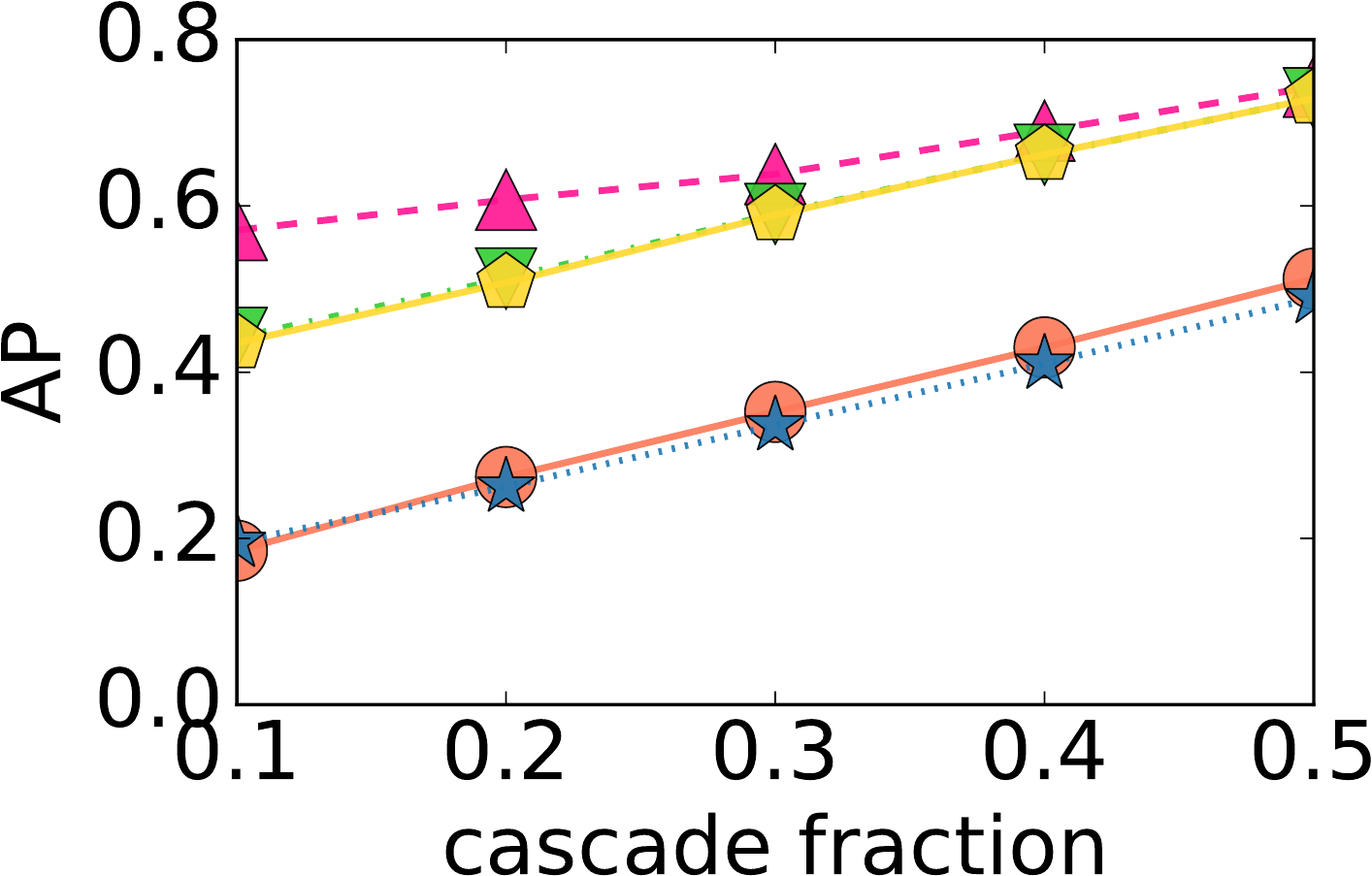} \\
    \includefig{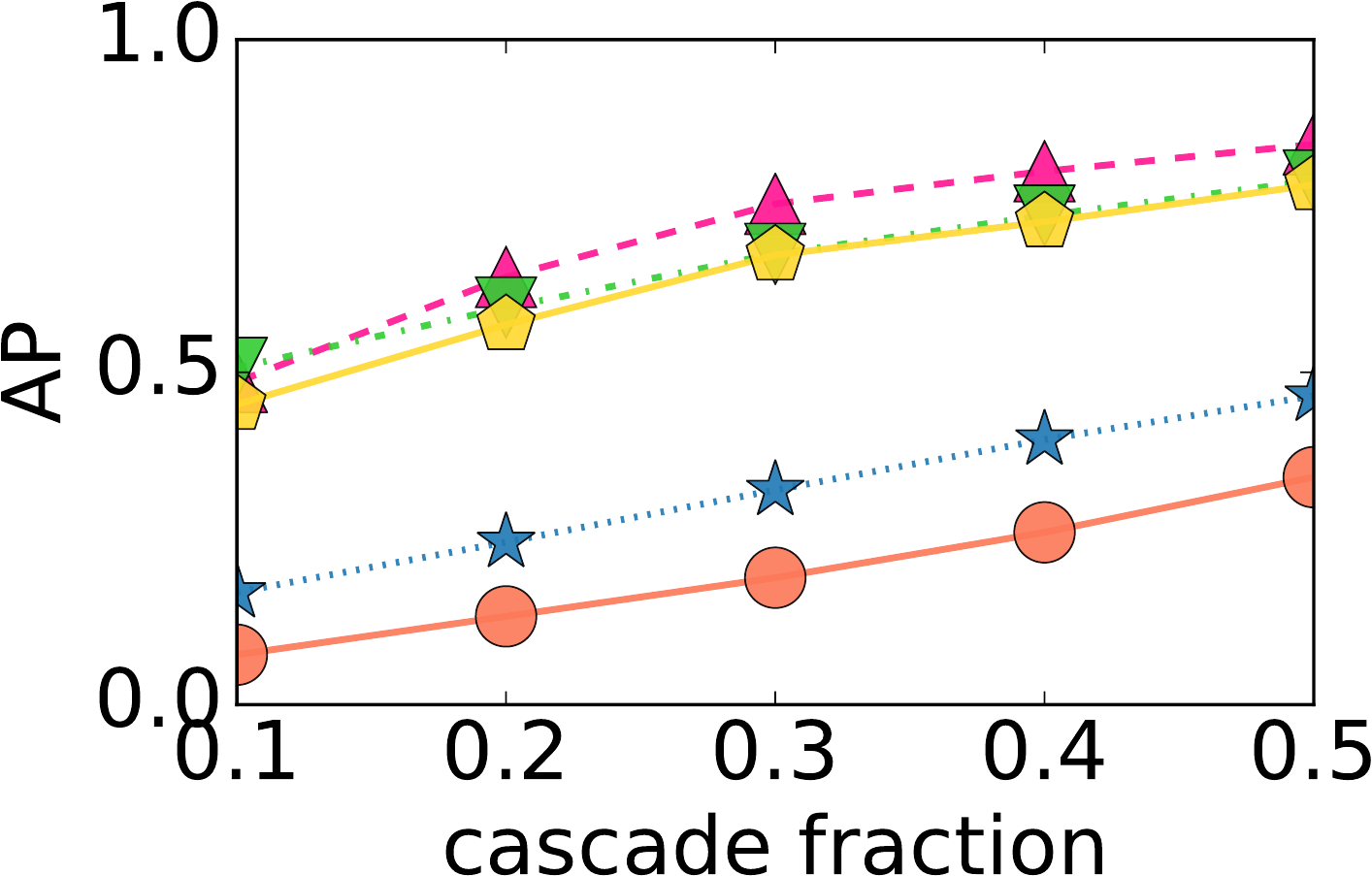} & \includefig{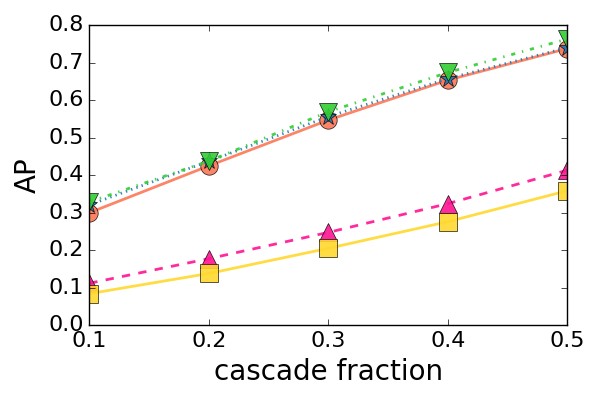} & \includefig{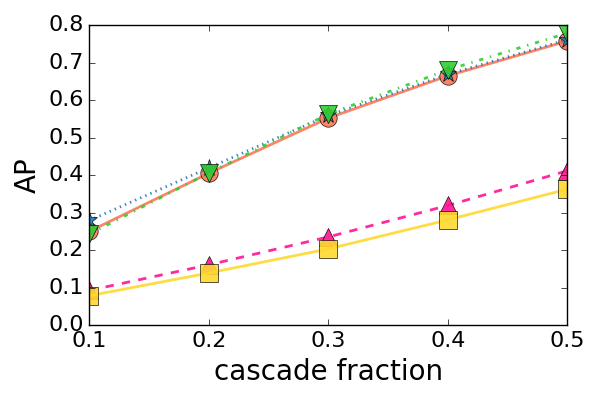} & \includefig{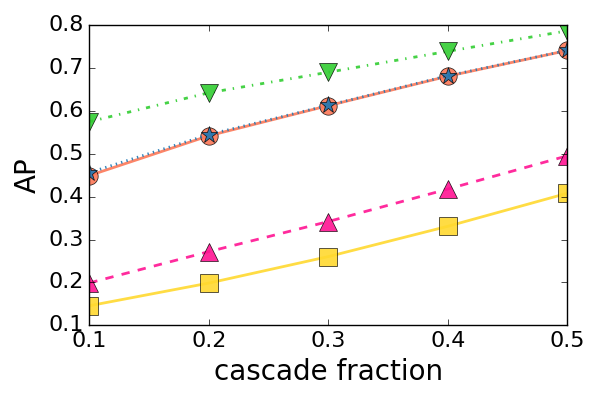} \\
    \multicolumn{4}{c}{\includelegend{{method_legend.pdf}}} 
  \end{tabular}  
  \caption{Node-level \ap score with respect to cascade size: \textbf{top row} (\ic model), \textbf{bottom row} (\si model)}
  \label{fig:node-by-cas}
\end{figure*}

\spara{Discussion on \grqc graph.}
On the \grqc graph, \pagerank performs better than \ourmethod with \trueroot while this is not the case in other graphs. 
We now provide some remarks to explain this difference in performance for the \grqc graph. 
First of all, we observe that \grqc has a very high assortativity coefficient~\cite{newman2002assortative} (AC) compared to the other graphs (as shown in Table~\ref{tbl:dataset}).
In graphs with large assortativity coefficient, high degree nodes tend to connect to each other, which implies that, the infected nodes in simulated cascades tend to form densely connected subgraphs. Notice that the random walker of \pagerank tends to give higher scores to the cascade subgraph if it's densely connected, which gives prediction advantage to \pagerank. In contrast, if the graph has low AC, the cascade subgraph tends to be sparsely connected, making \pagerank less effective. 

\subsection{Edge-level performance on synthetic cascades}
Next, we evaluate edge-level \ap scores under the same setting as the node-level evaluation.
Similarly, we consider two variants of \ourmethod with \trueroot or \mindistroot, 
while here we focus on the \ic model --- \si model yields similar results. 
Note that in this case, \netfill and \pagerank are not applicable by design because neither produces prediction on the edges. 

Shown in Figure~\ref{fig:edge}, \ourmethod, regardless of root selection strategy, outperforms \mintree.
Also, we observe that the performance increases both with the observation fraction and the cascade fraction. 
Meanwhile, the effect of root selection does not play an important role as we observe that the performance with \ourmethod with \mindistroot or \trueroot are very close. 

\begin{figure*}[t]
  \centering  
  \newcommand{\includefigobs}[1]{\includegraphics[width=0.24\linewidth]{figs/edge-cmp-baseline-obs-fractions/#1.pdf}}
  \newcommand{\includefigcas}[1]{\includegraphics[width=0.24\linewidth]{figs/edge-cmp-baseline-cascade-fractions/#1.pdf}}
  \newcommand{\includelegend}[1]{\includegraphics[width=0.8\linewidth]{figs/#1}}
  \setlength{\figwidth}{4.0cm}
  \setlength{\tabcolsep}{3pt}
  \newcolumntype{D}{ >{\centering\arraybackslash} m{1cm} }
  \begin{tabular}{cccc}
    \infectious  & \email & \student & \grqc \\
    \includefigobs{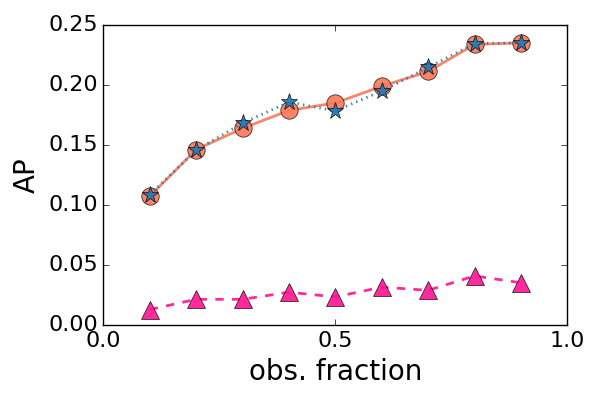} & \includefigobs{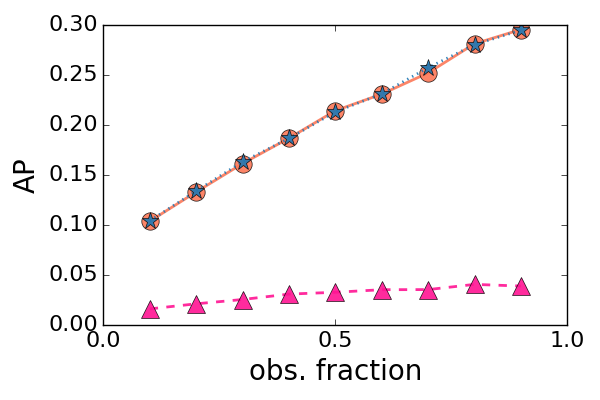} & \includefigobs{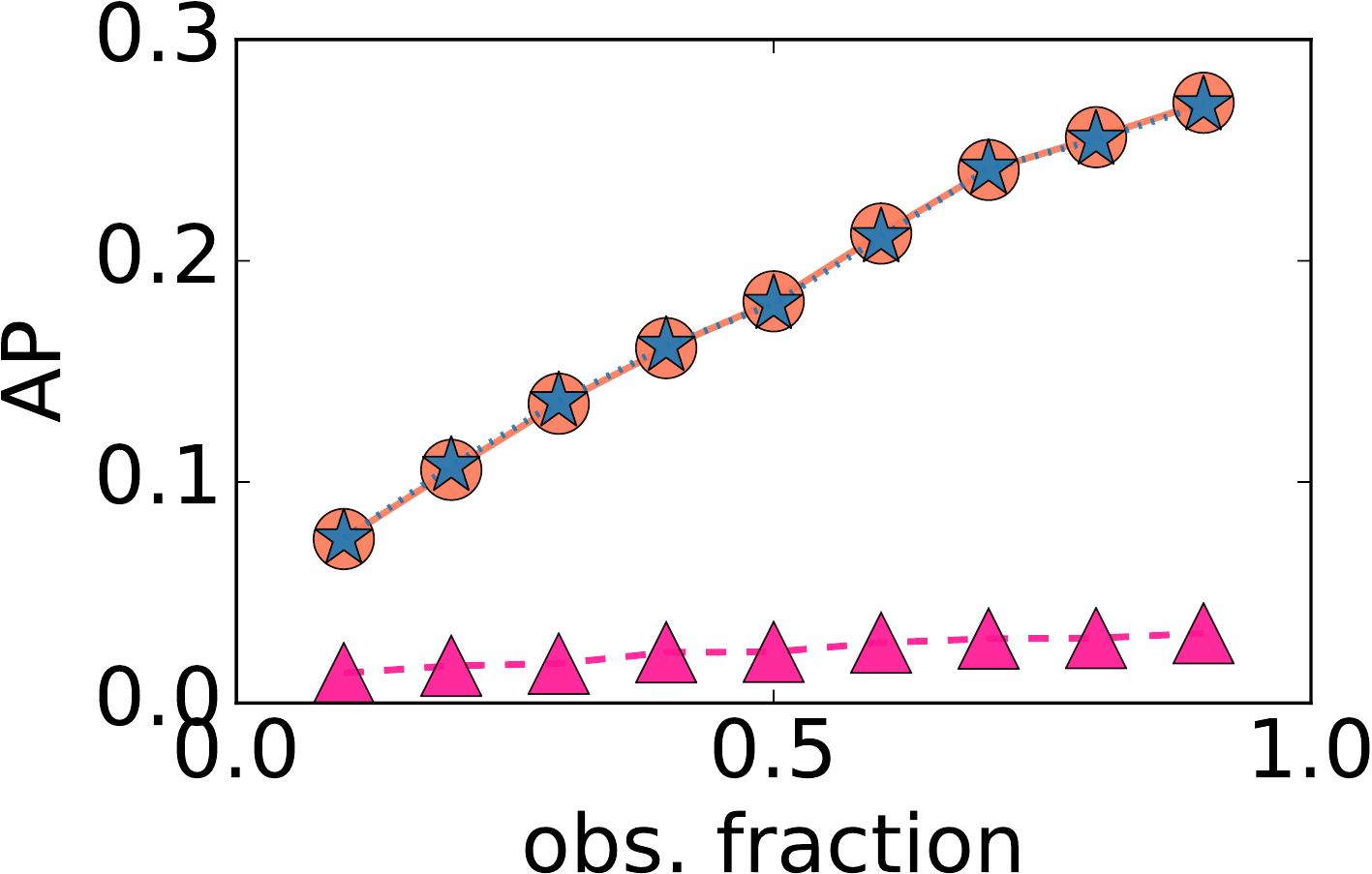} & \includefigobs{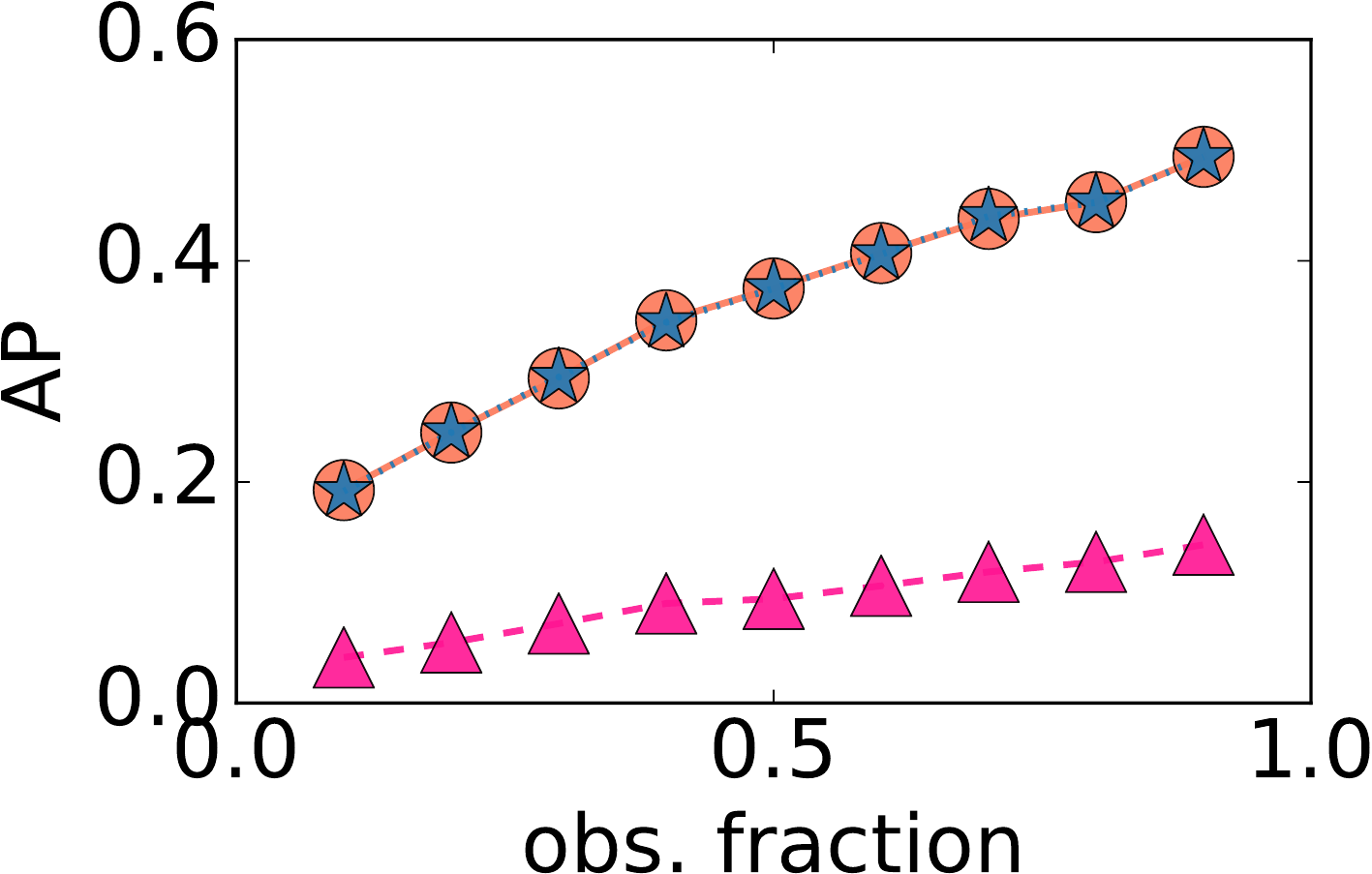} \\
    \includefigcas{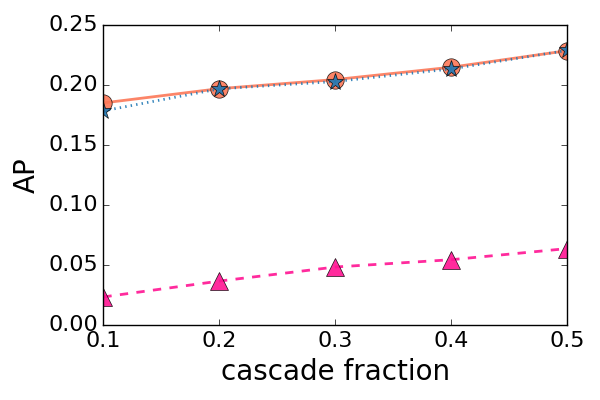} & \includefigcas{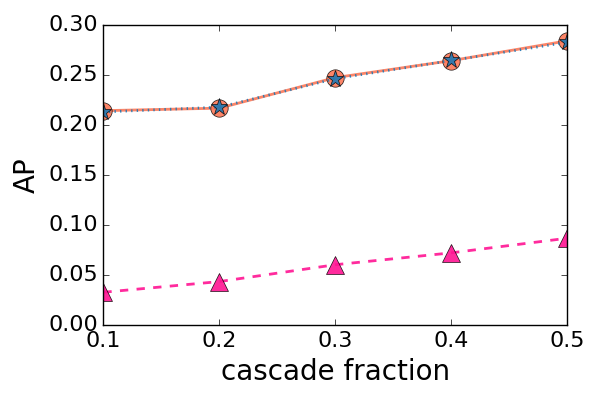} & \includefigcas{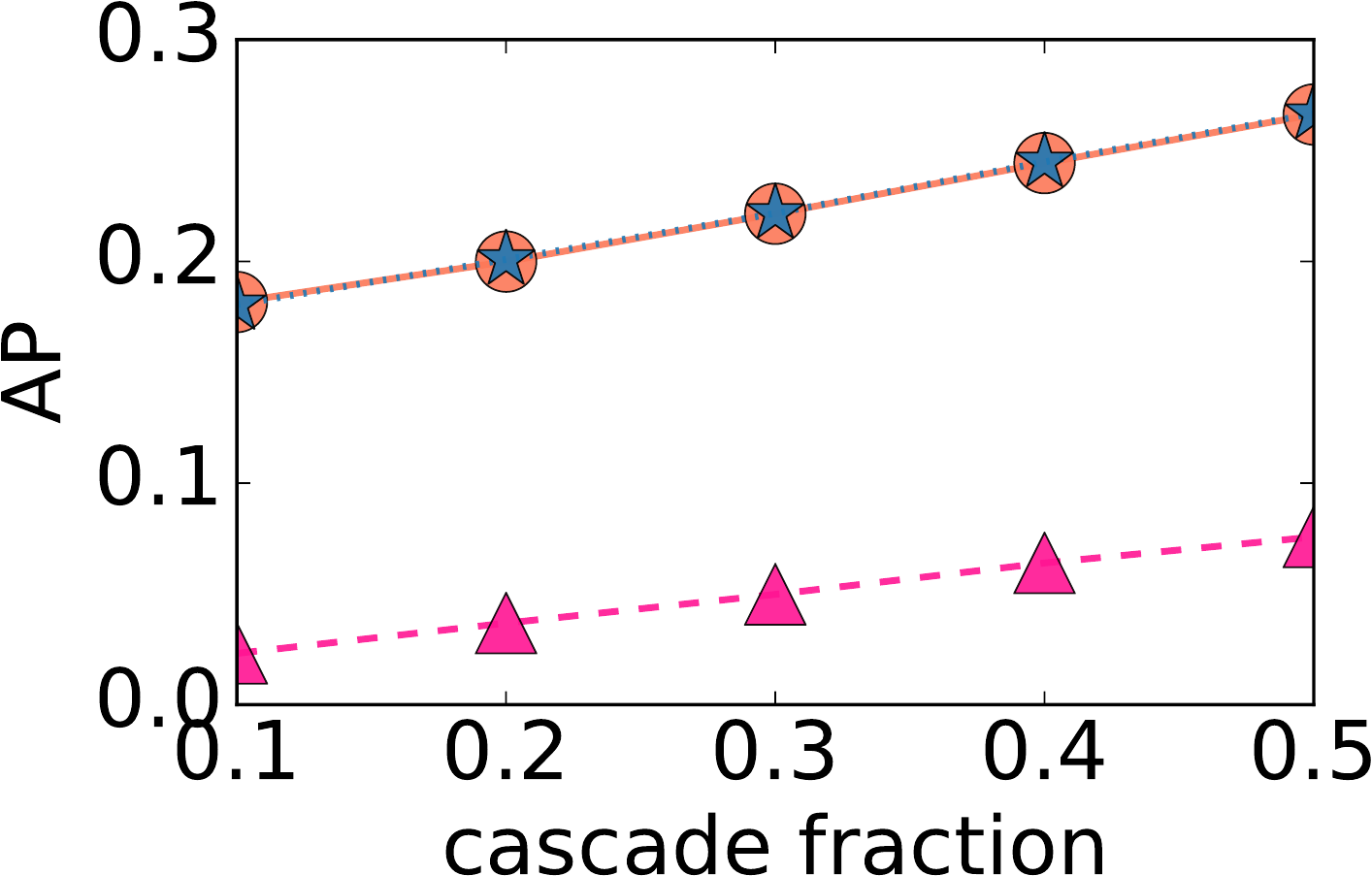} & \includefigcas{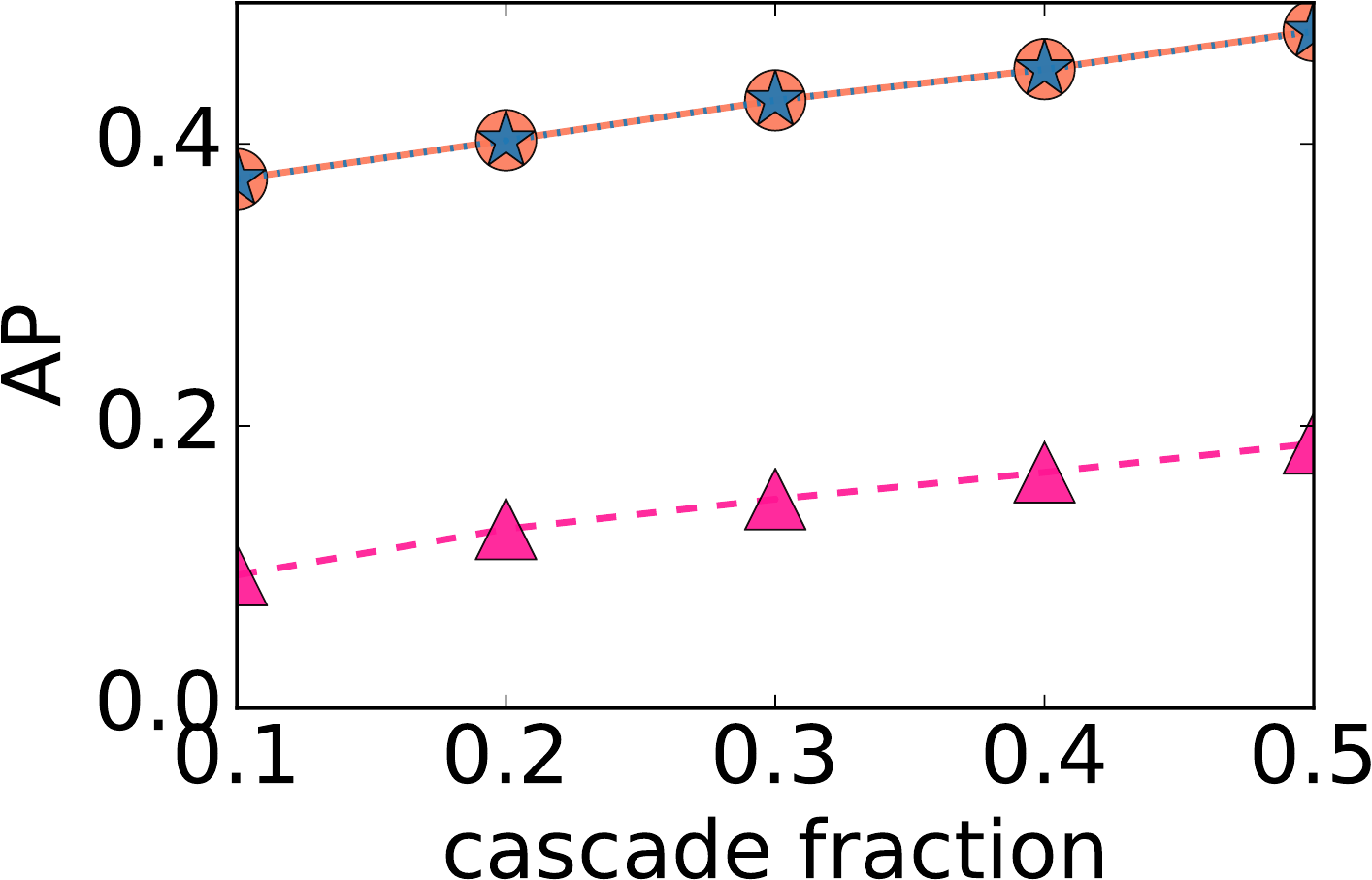} \\
    \multicolumn{4}{c}{\includelegend{{method_legend_edge.pdf}}} 
  \end{tabular}  
  \caption{Edge-level \ap score on \ic model with respect to the observation fraction (\textbf{top row}) and cascade fraction (\textbf{bottom row})}
  \label{fig:edge}
\end{figure*}

\subsection{Real cascades}
Figure~\ref{fig:digg} shows the experiment results on \digg dataset.%
\footnote{\netfill was not able to terminate within reasonable amount of time when using the implementation provided by the authors.}
We vary the observation fraction from $0.1$ to $0.8$ to recover only the hidden nodes since the dataset does not provide the edge-level propagation information. 
We observe that in \digg dataset, the \ap scores of all the methods are much lower than what is obtained for synthetic cascades. 
We also make the following observation:
first, the \ap scores of all methods are much lower than synthetic cascades in above sections. The reason can be two-fold: (1) the underlying infection propagation probabilities on the edges are unknown;
(2) the size of the cascade is very small (on average only 0.7\% of the whole graph is infected), making the prediction task much harder. 
We also observe that, \ourmethod (both with \pagerankroot and \trueroot) perform the best in such difficult setting, followed by \pagerank, and \mintree which performs the worst. 
Notice that \digg has assortativity coefficient close to zero and the fact \ourmethod outperforms \pagerank supports our remarks provided in Section~\ref{subsec:node-level}.


\begin{figure}[t]
  \centering  
  \includegraphics[width=0.6\linewidth]{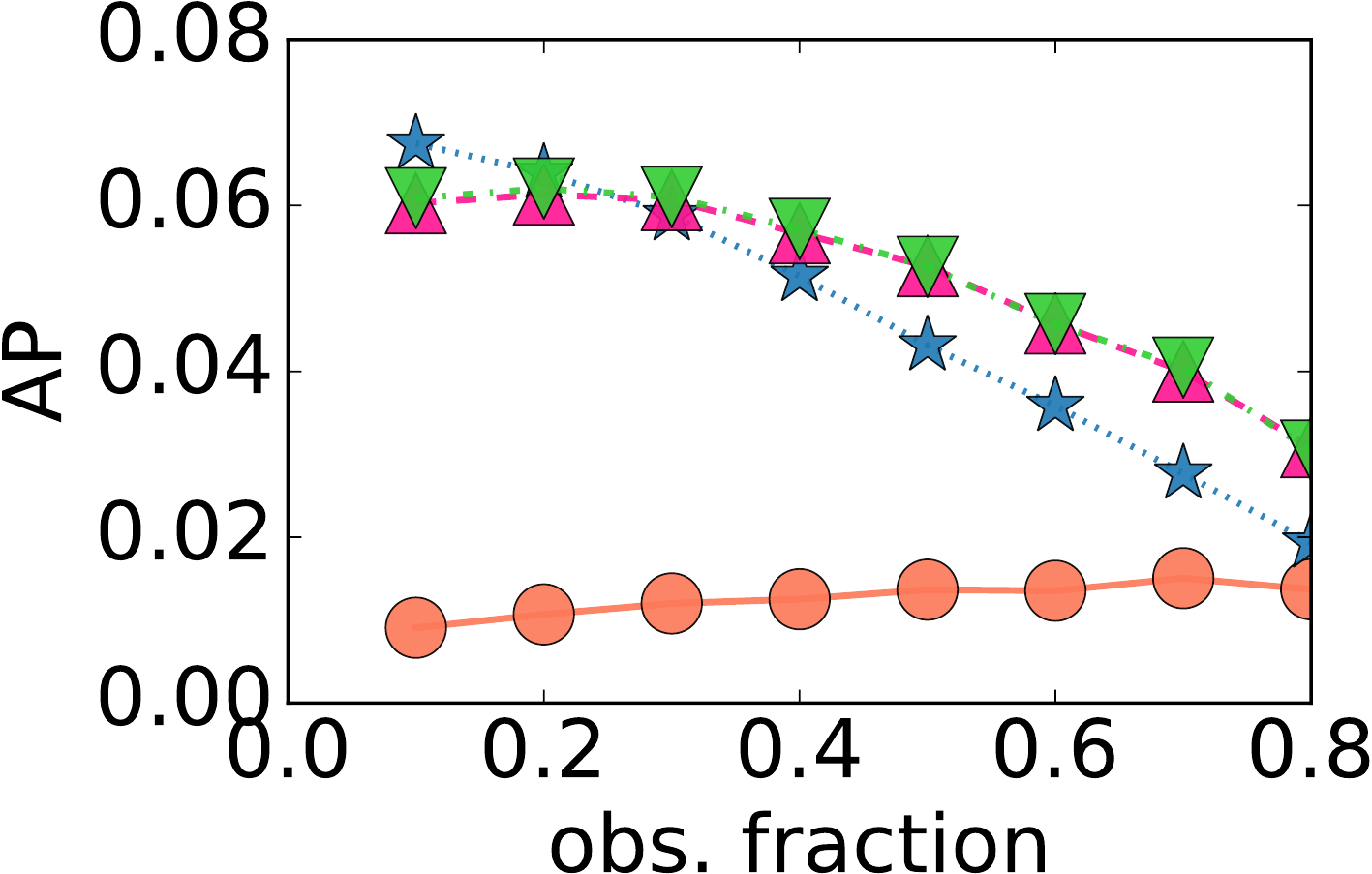} \\
  \vspace{0.5em}
  \includegraphics[width=0.66\linewidth]{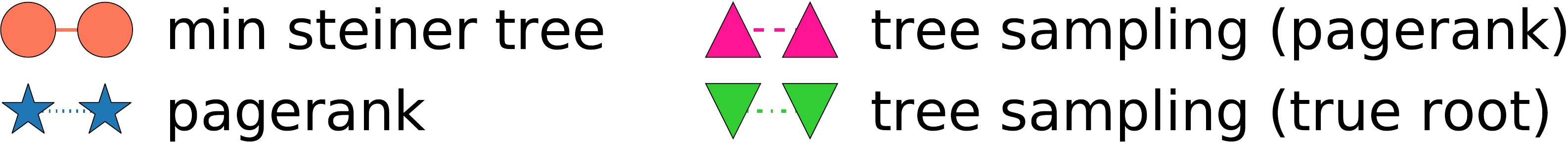}
  \caption{Node-level \ap score with respect to observation fraction on real cascades from \digg}
  \label{fig:digg}
\end{figure}

%% file: conclusion.tex
\section{Conclusion}

We study the problem of cascade-reconstruction in the probabilistic setting.
To estimate node infection probability, we reduce our problem to sampling Steiner trees.
We propose two novel algorithms with provable guarantees 
on the distribution of the sampled Steiner trees.
The proposed reconstruction algorithm makes fewer assumptions (e.g., cascade model and observation fraction)
compared to previous work. 
This makes our approach more robust in practice.
Experimental results show that the proposed approach outperforms the other baselines on most cases. 

Our work opens interesting directions for future research. 
First, it is important to re-visit Problem~\ref{prob:rooted} and design sampling algorithms that sample \textit{directly} from $\prod_{(u, v) \in \tree} p_{uv}$. 
Second, it is interesting to consider the problem of reconstructing cascades in the presence of temporal information, 
where we need to sample Steiner trees with node-order constraints.

%% file: DM834.bbl
\begin{thebibliography}{10}
\providecommand{\url}[1]{#1}
\csname url@samestyle\endcsname
\providecommand{\newblock}{\relax}
\providecommand{\bibinfo}[2]{#2}
\providecommand{\BIBentrySTDinterwordspacing}{\spaceskip=0pt\relax}
\providecommand{\BIBentryALTinterwordstretchfactor}{4}
\providecommand{\BIBentryALTinterwordspacing}{\spaceskip=\fontdimen2\font plus
\BIBentryALTinterwordstretchfactor\fontdimen3\font minus
  \fontdimen4\font\relax}
\providecommand{\BIBforeignlanguage}[2]{{%
\expandafter\ifx\csname l@#1\endcsname\relax
\typeout{** WARNING: IEEEtran.bst: No hyphenation pattern has been}%
\typeout{** loaded for the language `#1'. Using the pattern for}%
\typeout{** the default language instead.}%
\else
\language=\csname l@#1\endcsname
\fi
#2}}
\providecommand{\BIBdecl}{\relax}
\BIBdecl

\bibitem{rozenshtein2016reconstructing}
P.~Rozenshtein, A.~Gionis, B.~A. Prakash, and J.~Vreeken, ``Reconstructing an
  epidemic over time,'' in \emph{KDD}.\hskip 1em plus 0.5em minus 0.4em\relax
  ACM, 2016, pp. 1835--1844.

\bibitem{braunstein2016network}
A.~Braunstein and A.~Ingrosso, ``Network reconstruction from infection
  cascades,'' \emph{arXiv preprint arXiv:1609.00432}, 2016.

\bibitem{sundareisan2015hidden}
S.~Sundareisan, J.~Vreeken, and B.~A. Prakash, ``Hidden hazards: Finding
  missing nodes in large graph epidemics,'' in \emph{SDM}.\hskip 1em plus 0.5em
  minus 0.4em\relax SIAM, 2015, pp. 415--423.

\bibitem{xiao2018reconstructing}
H.~Xiao, P.~Rozenshtein, N.~Tatti, and A.~Gionis, ``Reconstructing a cascade
  from temporal observations,'' \emph{SDM}, 2018.

\bibitem{wilson1996generating}
D.~B. Wilson, ``Generating random spanning trees more quickly than the cover
  time,'' in \emph{Proceedings of the twenty-eighth annual ACM symposium on
  Theory of computing}.\hskip 1em plus 0.5em minus 0.4em\relax ACM, 1996, pp.
  296--303.

\bibitem{lappas2010finding}
T.~Lappas, E.~Terzi, D.~Gunopulos, and H.~Mannila, ``Finding effectors in
  social networks,'' in \emph{KDD}.\hskip 1em plus 0.5em minus 0.4em\relax ACM,
  2010, pp. 1059--1068.

\bibitem{shah:11:culprit}
D.~Shah and T.~Zaman, ``Rumors in a network: Who's the culprit?'' \emph{IEEE},
  vol.~57, no.~8, pp. 5163--5181, 2011.

\bibitem{farajtabar2015back}
M.~Farajtabar, M.~G. Rodriguez, M.~Zamani, N.~Du, H.~Zha, and L.~Song, ``Back
  to the past: Source identification in diffusion networks from partially
  observed cascades,'' in \emph{Artificial Intelligence and Statistics}, 2015,
  pp. 232--240.

\bibitem{mosbah1999non}
M.~Mosbah and N.~Saheb, ``Non-uniform random spanning trees on weighted
  graphs,'' vol. 218, no.~2.\hskip 1em plus 0.5em minus 0.4em\relax Elsevier,
  1999, pp. 263--271.

\bibitem{broder1989generating}
A.~Broder, ``Generating random spanning trees,'' in \emph{Foundations of
  Computer Science, 1989., 30th Annual Symposium on}.\hskip 1em plus 0.5em
  minus 0.4em\relax IEEE, 1989, pp. 442--447.

\bibitem{propp1998get}
J.~G. Propp and D.~B. Wilson, ``How to get a perfectly random sample from a
  generic markov chain and generate a random spanning tree of a directed
  graph,'' \emph{Journal of Algorithms}, vol.~27, no.~2, pp. 170--217, 1998.

\bibitem{hauptmann2015compendium}
M.~Hauptmann and M.~Karpinski, ``A compendium on steiner tree problems,'' 2015.

\bibitem{kempe2003maximizing}
D.~Kempe, J.~Kleinberg, and {\'E}.~Tardos, ``Maximizing the spread of influence
  through a social network,'' in \emph{Proceedings of the ninth ACM SIGKDD
  international conference on Knowledge discovery and data mining}, 2003, pp.
  137--146.

\bibitem{valiant1979complexity}
L.~G. Valiant, ``The complexity of enumeration and reliability problems,''
  vol.~8, no.~3.\hskip 1em plus 0.5em minus 0.4em\relax SIAM, 1979, pp.
  410--421.

\bibitem{tears}
A.~F. Smith and A.~E. Gelfand, ``Bayesian statistics without tears: a
  sampling--resampling perspective,'' \emph{The American Statistician},
  vol.~46, no.~2, pp. 84--88, 1992.

\bibitem{tutte1984graph}
W.~T. Tutte, ``Graph theory, volume 21 of encyclopedia of mathematics and its
  applications,'' 1984.

\bibitem{bahmani2010fast}
B.~Bahmani, A.~Chowdhury, and A.~Goel, ``Fast incremental and personalized
  pagerank,'' \emph{VLDB}, vol.~4, no.~3, pp. 173--184, 2010.

\bibitem{newman2002assortative}
M.~E. Newman, ``Assortative mixing in networks,'' \emph{Physical review
  letters}, vol.~89, no.~20, p. 208701, 2002.

\end{thebibliography}
